\documentclass [11pt, a4paper] {article}
\pdfoutput=1
\usepackage{fullpage}
\newcommand {\cCqyusRS} {Multi-Prover Quantum Merlin-Arthur Proof Systems with Small Gap}
\usepackage{amssymb}
\usepackage{amsfonts}
\usepackage{bbm}
\usepackage{hyperref}
\usepackage{bookmark}
\usepackage{amsthm}
\usepackage{mathtools}
\usepackage{algorithmic}
\usepackage{algorithm}

\algsetup{indent=2em}

\newcommand {\ZNbYlnmP} {\REQUIRE}
\newcommand {\tkOoHzex} {\ENSURE}

\newcommand {\ouhvXczx} {\usepackage[T1]{fontenc} \usepackage{mathptmx} \usepackage[scaled]{helvet} \usepackage{courier}}
\newcommand {\VCaaPKmo} [1] {\href{mailto:#1}{\texttt{#1}}}
\newcommand {\YALAmuJD} {Attila Pereszl\'{e}nyi}
\newcommand {\jpqbZwYu} {\VCaaPKmo{attila.pereszlenyi@gmail.com}}
\newcommand {\HQONrXtZ} {Centre for Quantum Technologies, National University of Singapore}
\newcommand {\cuIynkaR} {Without loss of generality}
\newcommand {\bmZlBfAS} {without loss of generality}
\newcommand {\iZJgfzde} {if and only if}
\newcommand {\ygMPbVVP} {such that}
\newcommand {\qpcpnInM} {I.e.,}
\newcommand {\WpedDmmP} {i.e.,}
\newcommand {\eRAJAvVL} {e.g.,}
\newcommand {\vHtwIXhr} {Cauchy--Schwarz inequality}
\newcommand {\fKqDCezI} {Chernoff bound}
\newcommand {\QgZuCZdC} {et al.}
\newcommand {\XaOdbtsT} {quantum interactive proof system}
\newcommand {\PmUQSIgQ} {\author {\YALAmuJD\thanks{E-mail: \jpqbZwYu.}\\ \textsl{\small \HQONrXtZ}}}
\newcommand {\LeqcErMT} {\bibliographystyle{halpha} \bibliography{./bib}}
\newcommand {\MJqbMYrm} [1] {Theorem~\ref{#1}}
\newcommand {\dWQdKYxc} [1] {Lemma~\ref{#1}}
\newcommand {\dOYmqxlU} [1] {Corollary~\ref{#1}}
\newcommand {\MOVPSHbF} [1] {Definition~\ref{#1}}
\newcommand {\FxucuSoM} [1] {Algorithm~\ref{#1}}
\newcommand {\HVRLKXNY} [1] {Section~\ref{#1}}
\newcommand {\IhApwfvN} [1] {Appendix~\ref{#1}}
\newcommand {\DnoYoBfx} [1] {\ensuremath{ \left( #1 \right) }}
\newcommand {\BKbTllLo} [1] {\ensuremath{ \left[ #1 \right] }}
\newcommand {\GGfSpoIn} [1] {\ensuremath{ \left\lbrace #1 \right\rbrace }}
\newcommand {\tXuOGzXr} {\: \! \!}
\newcommand {\HxLQfJAj} [2] {\ensuremath{ #1 \tXuOGzXr \DnoYoBfx{ #2 } }}
\newcommand {\QgeRYsKQ} [2] {\ensuremath{ #1 \tXuOGzXr \BKbTllLo{ #2 } }}
\DeclareMathOperator*{\argmax}{arg\,max}
\DeclareMathOperator*{\opE}{\mathbb{E}}
\newcommand {\faJHPjTb} {\ensuremath{ \stackrel {\mathrm{def}} {=} }}
\newcommand {\WHYFCIFc} {\ensuremath{ \iota }}
\newcommand {\QUPOtzml} [1] {\HxLQfJAj {O} {#1}}
\newcommand {\CNzHFRry} [1] {\HxLQfJAj {\widetilde{O}} {#1}}
\newcommand {\uaivIgOh} [1] {\HxLQfJAj {\Omega} {#1}}
\newcommand {\FIqGvTZG} [1] {\QUPOtzml{\log #1}}
\newcommand {\iteYiYHH} [1] {\GGfSpoIn{ 1, 2, \dotsc, #1 }}
\newcommand {\HelzWRDf} [1] {\GGfSpoIn{ 0, 1, \dotsc, #1 - 1 }}
\newcommand {\lUlEWeUc} [2] {\GGfSpoIn{ #1 : #2 }}
\newcommand {\NpwuHusD} [1] {\ensuremath{ \left\lceil #1 \right\rceil }}
\newcommand {\FyOFSCmf} [3] {\ensuremath{ #1 : \, #2 \rightarrow #3 }}
\newcommand {\ipeoICuS} {\ensuremath{ \mathsf{poly} }}
\newcommand {\wrDoIJTw} [1] {\HxLQfJAj {\ipeoICuS} {#1}}

\newcommand {\jVGfTsvf} [1] {\ensuremath{ \mathbb{#1} }}
\newcommand {\VWqEPttf} {\jVGfTsvf{C}}
\newcommand {\VrXQeMta} {\jVGfTsvf{N}}
\newcommand {\CYLQkbAO} {\ensuremath{ \jVGfTsvf{Z}^{+} }}
\newcommand {\fHAwcqZo} [1] {\ensuremath{ \left\langle #1 \right| }}
\newcommand {\raIfJZxu} [1] {\ensuremath{ \left| #1 \right\rangle }}
\newcommand {\ZFdmJtwX} [2] {\ensuremath{ \raIfJZxu{#1} \! \fHAwcqZo{#2} }}
\newcommand {\GyFczGjW} [1] {\ZFdmJtwX{#1}{#1}}
\newcommand {\NChmtXsR} [2] {\ensuremath{ \left\langle #1 \middle\vert #2 \right\rangle }}

\newcommand {\NUxdODWZ} {\ensuremath{ \otimes }}
\newcommand {\pTaTOfgX} {\ensuremath{ \mathrm{Tr} }}
\newcommand {\AfBMaxZI} [1] {\ensuremath{ \left| #1 \right| }}
\newcommand {\GoyqnUbk} [1] {\ensuremath{ \left\| #1 \right\| }}
\newcommand {\cbKdRoIN} [2] {\ensuremath{ \GoyqnUbk{#1}_{#2} }}
\newcommand {\BdlsvLWK} [1] {\cbKdRoIN{#1}{\pTaTOfgX}}
\newcommand {\xOJdCehH} [1] {\cbKdRoIN{#1}{1}}
\newcommand {\vPiROeRG} [2] {\ensuremath{ \frac{1}{2} \BdlsvLWK{#1 - #2} }}
\newcommand {\TdSeQgKh} [1] {\QgeRYsKQ{\Pr}{#1}}
\newcommand {\YcVNkVQD} [1] {\QgeRYsKQ{\opE}{#1}}
\newcommand {\yxpqVGXT} [1] {\ensuremath{ \mathsf{#1} }}
\newcommand {\OduTrvoO} [1] {\ensuremath{\text{\textsc{#1}}}}
\newcommand {\hpRexCjN} {\OduTrvoO{3Sat}}
\newcommand {\OktNAkJL} {\yxpqVGXT{NP}}
\newcommand {\LHqjxWJM} {\yxpqVGXT{PSPACE}}
\newcommand {\QZccOLoE} {\yxpqVGXT{EXP}}
\newcommand {\vxJDyEiW} {\yxpqVGXT{NEXP}}
\newcommand {\FmrmOQaW} {\yxpqVGXT{IP}}
\newcommand {\kJmbkUXz} {\yxpqVGXT{MA}}
\newcommand {\sVwdElyU} {\yxpqVGXT{PP}}
\newcommand {\SHvwttaw} {\yxpqVGXT{QMA}}
\newcommand {\ACroxxZA} {\yxpqVGXT{QCMA}}
\newcommand {\cdToKcBp} {\yxpqVGXT{BellQMA}}
\newcommand {\ANJFILuJ} {\yxpqVGXT{LOCCQMA}}
\newcommand {\JWevirQe} {\yxpqVGXT{QIP}}
\newcommand {\VGZgApSi} {\yxpqVGXT{BQP}}
\newcommand {\AIJjtsxK} {\yxpqVGXT{PQP}}
\newcommand {\KOzxJkVZ} [4] {\HxLQfJAj{\SHvwttaw}{#1, #2, #3, #4}}
\newcommand {\fiWpgtYc} [4] {\HxLQfJAj{\cdToKcBp}{#1, #2, #3, #4}}
\newcommand {\lrcGeBqo} [1] {\QgeRYsKQ{\SHvwttaw}{#1}}
\newcommand {\pUeVeClC} [1] {\QgeRYsKQ{\cdToKcBp}{#1}}
\newcommand {\GOljMfAR} [1] {\QgeRYsKQ{\ANJFILuJ}{#1}}
\newcommand {\vtZmVJmM} [2] {\HxLQfJAj{\ACroxxZA}{#1, #2}}
\newcommand {\wQoFYwHz} [1] {\ensuremath{ \mathnormal{#1} }}
\newcommand {\nVzyhtCm} [1] {\ensuremath{ \mathbf{#1} }}
\newcommand {\gctqElrH} {\nVzyhtCm{CNOT}}
\newcommand {\lDGygnkY} {\nVzyhtCm{H}}
\newcommand {\sYcexCeY} [2] {\HxLQfJAj {\nVzyhtCm{R}_{#1}} {#2}}
\newcommand {\YMElKjII} {\ensuremath{ \mathbbm{1} }}
\newcommand {\VSdxZJNO} [1] {\ensuremath{ \mathsf{#1} }}
\newcommand {\zdLEuPyh} [1] {\ensuremath{ \mathcal{#1} }}
\newcommand {\PZzafWXr} [1] {\ensuremath{\mathrm{#1}}}
\newcommand {\KKIedNYP} [1] {\HxLQfJAj{\PZzafWXr{D}}{#1}}
\theoremstyle {plain}
\newtheorem {thm} {Theorem} [section]
\newtheorem {cor} [thm] {Corollary}
\newtheorem {lem} [thm] {Lemma}
\theoremstyle {remark}
\newtheorem {rem} [thm] {Remark}

\theoremstyle {definition}
\newtheorem {defi} [thm] {Definition}

\hypersetup{
pdftitle={\cCqyusRS},
pdfauthor={\YALAmuJD},
pdfstartview={FitH},
breaklinks={true}
}
\ouhvXczx
\newcommand {\HNGSAaAW} {\OduTrvoO{Succinct3Col}}
\newcommand {\oVrArWkh} {\ensuremath{C_G}}
\title {\textbf{\cCqyusRS}}
\PmUQSIgQ
\date {May 11, 2012}
\begin{document}
\maketitle
\begin{abstract}
This paper studies multiple-proof quantum
Merlin-Arthur (\SHvwttaw) proof systems in the setting
when the completeness-soundness gap is small.
Small means that we only lower-bound the gap
with an inverse-exponential function of the
input length, or with an even smaller function.
Using the protocol of Blier and Tapp \cite{Blier2009},
we show that in this case the proof system
has the same expressive power as non-deterministic
exponential time (\vxJDyEiW).
Since single-proof \SHvwttaw{} proof systems,
with the same bound on the gap,
have expressive power at most exponential time (\QZccOLoE),
we get a separation between single and
multi-prover proof systems in the `small-gap setting',
under the assumption that $ \QZccOLoE \neq \vxJDyEiW $.
This implies, among others, the nonexistence of
certain operators called disentanglers
(defined by Aaronson \QgZuCZdC\ \cite{Aaronson2009}),
with good approximation parameters.
\par
We also show that in this setting the proof system
has the same expressive power if we restrict the verifier
to be able to perform only Bell-measurements, \WpedDmmP{}
using a \cdToKcBp{} verifier.
This is not known to hold in the usual setting,
when the gap is bounded by an inverse-polynomial
function of the input length.
To show this we use the protocol of Chen and Drucker
\cite{Chen2010}.
The only caveat here is that we need at least a linear
amount of proofs to achieve the power of \vxJDyEiW,
while in the previous setting two
proofs were enough.
\par
We also study the case when the proof-lengths are
only logarithmic in the input length and observe
that in some cases the expressive power decreases.
However, we show that it doesn't decrease further if we make
the proof lengths to be even shorter.
\end{abstract}
\section{Introduction}
Arthur-Merlin games and the class \kJmbkUXz{} was defined by
Babai \cite{Babai1985} as natural extension of the class \OktNAkJL{}
using randomization.
In the definition of \kJmbkUXz{} the prover (Merlin) gives a polynomial
length `proof' to the verifier (Arthur), who then performs a
polynomial-time randomized computation and has to decide
if an input $x$ is in a language or not.
The verifier is allowed to make some error in the decision,
hence making the class \kJmbkUXz{} possibly more powerful than \OktNAkJL.
If we add communication to the model, \WpedDmmP{} the prover and
the verifier can exchange a polynomial number of messages then
we get the class \FmrmOQaW{} \cite{Goldwasser1989}.\footnote{Babai also defined
an interactive version of \kJmbkUXz, that can be thought of as a
`public-coin' version of \FmrmOQaW.
Later Goldwasser and Sipser \cite{Goldwasser1986} showed
that this class has the same expressive power as \FmrmOQaW.}
The classes \FmrmOQaW{} and \kJmbkUXz{} has been extensively studied and
it is known that in both cases we can make the protocol to have
one-sided, exponentially small error without hurting the
power of the proof system.
\qpcpnInM{} the verifier only makes an error if the input is not
in the language, and even in this case the error
probability is at most an inverse-exponential
function of the input length.
Surprisingly, it turned out that \FmrmOQaW{} is equal to the class
of problems decidable in polynomial space (\LHqjxWJM)
\cite{Lund1992,Shamir1992}.
For more information on these classes see \eRAJAvVL{} the book
of Arora and Barak \cite{Arora2009}.
\par
Quantum Merlin-Arthur proof systems (and the class \SHvwttaw)
were introduced by Knill \cite{Knill1996}, Kitaev
\cite{Kitaev2002}, and also by Watrous \cite{Watrous2000}
as a natural extension
of \kJmbkUXz{} and \OktNAkJL{} to the quantum computational setting.
Similarly, \XaOdbtsT{}s (and the class \JWevirQe) were introduced by
Watrous \cite{Watrous2003} as a quantum
analogue of \FmrmOQaW.
These classes have also been well studied and now it's known
that the power of \XaOdbtsT{}s is the same as the classical ones,
\WpedDmmP{} $ \JWevirQe = \FmrmOQaW = \LHqjxWJM $ \cite{Jain2010}.
Furthermore, \XaOdbtsT{}s still have the same expressive power
if we restrict the number of messages to three and
have exponentially small one-sided error \cite{Kitaev2000}.
The class \SHvwttaw{} can also be made to have exponentially small error,
and has natural complete problems \cite{Aharonov2002}.
Interestingly it's not known if we can make \SHvwttaw{} to have
one-sided error.\footnote{In a recent paper, Jordan
\QgZuCZdC\ \cite{Jordan2011} showed that the proof system can achieve
perfect completeness if the prover's message is classical.}
\par
Several variants of \JWevirQe{} and \SHvwttaw{} have also been studied.
For example one can consider the case where some or
all of the messages are
short, meaning at most logarithmic in the input length
\cite{Marriott2005,Beigi2011,Pereszlenyi2011}.
Our focus will be more on the setting that was introduced by
Ito, Kobayashi and Watrous \cite{Ito2010}.
They studied quantum classes where the gap between the
completeness and soundness parameter is very small.
Their main result is that quantum interactive proofs with
double-exponentially small gap is characterized by
\QZccOLoE{} (deterministic exponential time).
\par
Probably the most interesting generalization of \SHvwttaw{}
is by Kobayashi, Matsumoto and Yamakami \cite{Kobayashi2003}
who defined the class \lrcGeBqo{k}.
In this setting there are $k$ provers who send $k$ quantum proofs
to the verifier, and these proofs are guaranteed to be unentangled.
Note that in the classical setting this generalization is not
interesting since we can just concatenate the $k$ proofs and
treat them as one proof.
However, in the quantum case a single prover can entangle
the $k$ proofs and no method is known to detect such
a cheating behavior.
\par
Obviously the most important question is whether more
provers make the class more powerful or not.
In a later version of their paper, Kobayashi \QgZuCZdC\ \cite{Kobayashi2003}
(and independently Aaronson \QgZuCZdC\ \cite{Aaronson2009})
showed that $ \lrcGeBqo{2} = \lrcGeBqo{k} $ for all
polynomially-bounded $k$ \iZJgfzde{} \lrcGeBqo{2} can be amplified
to exponentially small error.
Later Harrow and Montanaro \cite{Harrow2010} showed that the above equality
indeed holds.
The question now is whether \SHvwttaw{} is equal to \lrcGeBqo{2},
or in other words, does unentanglement actually help?
There are signs that show that the above two classes are
probably not equal.
For example, Liu, Christandl, and Verstraete \cite{Liu2007} found a problem
that has a \lrcGeBqo{2} proof system, but not known to belong
to \SHvwttaw.
Blier and Tapp \cite{Blier2009} showed that all problems in
\OktNAkJL{} have a \lrcGeBqo{2} proof system where the length of
both proofs are logarithmic in the input length.
On the other hand, if \SHvwttaw{} has one logarithmic-length proof
then it has the same expressive power as \VGZgApSi{} \cite{Marriott2005}.
Since \VGZgApSi{} is not believed to contain \OktNAkJL{},
\lrcGeBqo{2} with logarithmic length proofs is probably
more powerful than \SHvwttaw{} with a logarithmic proof.
The above proof system had some inverse-polynomial
gap, and this gap was later improved by several papers
\cite{Beigi2010,Chiesa2011,Gall2011}.
However, in all of these improvements the gap is still an
inverse-polynomial function of the input
length.\footnote{It is not believed that
the gap in this setting can be improved to a constant
because it would imply that $ \lrcGeBqo{2} = \vxJDyEiW $. \cite{Aaronson2009}}
Another evidence is by Aaronson \QgZuCZdC\ \cite{Aaronson2009}
who found a \lrcGeBqo{\CNzHFRry{\sqrt{n}}} proof
system for \hpRexCjN{} with constant gap
and where each proof consist of \FIqGvTZG{n} qubits.
Again, it seems unlikely that \hpRexCjN{}
has a proof system with one \CNzHFRry{\sqrt{n}}-length proof.
\subsection{Our Contribution}
We study multiple-proof \SHvwttaw{} proof systems in the setting
where the completeness-soundness gap is exponentially
small or even smaller.
We examine three variants of these proof systems
as described below.
\subsubsection{\texorpdfstring{\lrcGeBqo{k} with Small Gap}
{QMA[k] with Small Gap}}
The first variant we look at is the small-gap version of \lrcGeBqo{k}
mentioned above.
We show that this class is exactly characterized
by \vxJDyEiW{} if the number of proofs are between $2$ and
\wrDoIJTw{n}, and the completeness-soundness gap is between
exponentially or double-exponentially small.
The power of the proof system is still \vxJDyEiW{} if
we require it to have one-sided error.
More precisely we show the following theorem.
\begin{thm}
$ \displaystyle \vxJDyEiW =
\KOzxJkVZ{\ipeoICuS}{2}{1}{1 - \uaivIgOh{4^{-n}}} =
\bigcup_{\substack{0 < s < c \leq 1, \\ c-s \geq 2^{ -2^{\ipeoICuS} }}}
\KOzxJkVZ{\ipeoICuS}{\ipeoICuS}{c}{s} $,
where \HxLQfJAj{c}{n} and \HxLQfJAj{s}{n} can be calculated in time
at most exponential in $n$ on a classical computer.
\label{NmPOZGKI}
\end{thm}
In the notation above the first parameter of \SHvwttaw{} is
the upper-bound on the length of each proofs the verifier
receives, in qubits. The second parameter denotes
the number of unentangled proofs, while the third
is the completeness and the fourth is the soundness parameter.
For a precise definition of the above notation
see \HVRLKXNY{EriaXXKj}.
Note that in \MJqbMYrm{NmPOZGKI} the \vxJDyEiW{} upper-bound
is trivial, as it follows from exactly the same argument that
shows the \vxJDyEiW{} upper-bound to the `normal-gap' \lrcGeBqo{2}.
Interestingly, there is no other upper-bound known for \lrcGeBqo{2},
and it is a big open question to strengthen this bound
\cite{Aaronson2009}.
The surprising phenomenon is that if we relax the bound
on the gap, then the expressive power of the class jumps
all the way up to the trivial upper-bound.
Note that an \QZccOLoE{} upper-bound for the small-gap, single-prover
\SHvwttaw{} is easily seen, so we have a separation between
\SHvwttaw{} and \lrcGeBqo{k} in the small-gap setting.
For more discussion about this and other consequences see
\HVRLKXNY{HujYVLzw}.
\par
The non-trivial part of the proof is proving the \vxJDyEiW{}
lower-bound.
For this we use the protocol of Blier and Tapp \cite{Blier2009}
on a \vxJDyEiW-complete language which we call \HNGSAaAW,
the succinct version of graph 3-coloring.
The proof of \MJqbMYrm{NmPOZGKI} is
presented in \HVRLKXNY{jvxzezwT}.
\subsubsection{\texorpdfstring{\pUeVeClC{k} with Small Gap}
{BellQMA[k] with Small Gap}}
The class \pUeVeClC{k} was defined by
Aaronson \QgZuCZdC\ \cite{Aaronson2009},
Brand\~{a}o \cite{Brandao2008}, and
Chen and Drucker \cite{Chen2010}.
The above definitions are not exactly the same,
but the subtle difference doesn't matter in either
of the above papers, neither it does in this paper.
The exact definition of the class can be found
in \HVRLKXNY{EriaXXKj}.
Roughly, the difference between \lrcGeBqo{k} and
\pUeVeClC{k} is that in the latter the verifier
has to measure each proof separately and non-adaptively,
then based on the outcomes has to make its decision.
Aaronson \QgZuCZdC\ \cite{Aaronson2009} asked the question whether
\pUeVeClC{k} has the same power as \lrcGeBqo{k} and
if there is a \cdToKcBp{} protocol for \hpRexCjN{}
with similar parameters as theirs.
A partial positive answer to the first question was
given by Brand\~{a}o \cite{Brandao2008}, who showed that
$ \pUeVeClC{\QUPOtzml{1}} = \SHvwttaw $, and
a positive answer to the second question
was given by Chen and Drucker \cite{Chen2010}.
The power of \pUeVeClC{k} with super-constant $k$
remains an open problem.
\par
In this paper we study the small-gap version of
\pUeVeClC{k}, where again small means
exponentially or double-exponentially small.
One can observe that
Brand\~{a}o's proof of $ \pUeVeClC{\QUPOtzml{1}} = \SHvwttaw $
doesn't go through if the gap is so small.\footnote{Aaronson
\QgZuCZdC\ \cite{Aaronson2009} also defined
the class \ANJFILuJ{} similarly to
\cdToKcBp{} but allowing the verifier to make adaptive
and even several measurements on the same proofs.
Brand\~{a}o, Christandl, and Yard \cite{Brandao2011}
showed that $ \GOljMfAR{\QUPOtzml{1}} = \SHvwttaw $.
This proof also breaks down if the gap is small.}
So we don't know the power of \pUeVeClC{k} with
constant $k$ in the small-gap setting.
However, we show that if $k = \uaivIgOh{n}$
then \pUeVeClC{k} has the same power as
\lrcGeBqo{k}, \WpedDmmP{} it also equals to \vxJDyEiW.
This is expressed by the following theorem.
\begin{thm}
$ \displaystyle \vxJDyEiW =
\fiWpgtYc{\ipeoICuS}{\uaivIgOh{n}}{c}{s} =
\bigcup_{\substack{0 < s' < c' \leq 1, \\ c'-s' \geq 2^{ -2^{\ipeoICuS} }}}
\fiWpgtYc{\ipeoICuS}{\ipeoICuS}{c'}{s'} $,
for some $c$ and $s$ with
$ \HxLQfJAj{c}{n} - \HxLQfJAj{s}{n} = \uaivIgOh{4^{-n}} $, and where
\HxLQfJAj{c'}{n} and \HxLQfJAj{s'}{n} can be calculated in time
at most exponential in $n$ on a classical computer.
\label{nKVfzRHR}
\end{thm}
In the above the \vxJDyEiW{} upper-bound is again trivial,
so the only thing we need to do is give a \cdToKcBp{}
protocol for \vxJDyEiW.
Just as in the proof of the previous theorem,
we will use the same language (\HNGSAaAW) and
give a proof system for that.
For this we will use the protocol of
Chen and Drucker \cite{Chen2010}.
Note that similarly to their original
protocol we don't have
perfect completeness either.
\par
Knowing that having sufficiently many
provers makes \pUeVeClC{k} as powerful as
\lrcGeBqo{k}, it would be very interesting
to know the power of \pUeVeClC{k} with
small gap and constant $k$.
It's either more powerful then
$\pUeVeClC{1} = \SHvwttaw $, or somewhere between
constant and linear number of provers there is
an increase in the expressive power.
A more detailed discussion about this can be found
in \HVRLKXNY{fCJGhyBh}, and the proof of
\MJqbMYrm{nKVfzRHR} is presented in
\HVRLKXNY{BXwCHKEG}.
\subsubsection{\texorpdfstring{\lrcGeBqo{k} with Small Gap and Short Proofs}
{QMA[k] with Small Gap and Short Proofs}}
We also study the small-gap version of \lrcGeBqo{k}
where each proof is short \WpedDmmP{} at most
\FIqGvTZG{n}-many qubits.
One would expect that in this case the
class gets weaker then \vxJDyEiW.
We could only show this in the case when
the gap is inverse-exponential, which follows
from a few simple observations.
For the case when the gap is smaller than this
we give a lemma that simplifies the proof system,
which may help to prove non-trivial upper-bounds
for these classes later.
The lemma roughly says that instead of
\FIqGvTZG{n}-many qubits, we can consider the
case when all the proof are just $1$ qubits.
The above mentioned results are described in
more detail in \HVRLKXNY{sec:QMAlog}.
\subsection*{Organization of the Paper}
The remainder of the paper is organized as follows.
\HVRLKXNY{EriaXXKj} discusses the background
theorems and definitions needed for the rest of the paper.
\HVRLKXNY{jvxzezwT} presents the proof of
our first main theorem, \MJqbMYrm{NmPOZGKI},
with the last part of the proof presented
in \IhApwfvN{WEFVXHAK}.
In \HVRLKXNY{BXwCHKEG} we explain
the proof of our other main theorem,
\MJqbMYrm{nKVfzRHR}.
We end the paper with some conclusions and open
problems in \HVRLKXNY{HujYVLzw}, with
one last proof presented in
\IhApwfvN{zKhwWgDg}.
\section{Preliminaries}
\label{EriaXXKj}
We assume familiarity with quantum information and
computation \cite{Nielsen2000}; such as
quantum states, unitary operators, measurements, etc.
We also assume the reader is familiar with
computational complexity, both classical \cite{Arora2009}
and quantum \cite{Watrous2008a}.
Throughout the paper we will talk about complexity
classes like \sVwdElyU, \LHqjxWJM, \QZccOLoE, \vxJDyEiW,
\VGZgApSi, \AIJjtsxK, \SHvwttaw, and \JWevirQe.
In this section we only define the classes that
are the most relevant to our discussions,
where the definition of the rest can be found
in the above mentioned references.
The purpose of this section is to present some of the notations
and background information (definitions,
theorems) required to understand the rest of the paper.
\par
We denote the set of functions of $n$ that are
upper-bounded by some polynomial in $n$ by \wrDoIJTw{n}.
If the argument is clear, we omit it and just
write \ipeoICuS.
We denote the imaginary unit by \WHYFCIFc{} instead of
$i$, which we use as an index in summations for example.
\begin{defi}
In the paper we use some well-known quantum gates.
We define them here.
\begin{align*}
\gctqElrH &\faJHPjTb \GyFczGjW{00} + \GyFczGjW{01} +
\ZFdmJtwX{11}{10} + \ZFdmJtwX{10}{11} \\
\lDGygnkY &\faJHPjTb \frac{1}{\sqrt{2}}
\DnoYoBfx{\GyFczGjW{0} + \ZFdmJtwX{0}{1} + \ZFdmJtwX{1}{0} - \GyFczGjW{1} } \\
\sYcexCeY{x}{\omega} &\faJHPjTb \cos \frac{\omega}{2} \GyFczGjW{0}
-\WHYFCIFc \sin \frac{\omega}{2} \ZFdmJtwX{0}{1}
-\WHYFCIFc \sin \frac{\omega}{2} \ZFdmJtwX{1}{0}
+ \cos \frac{\omega}{2} \GyFczGjW{1} \\
\sYcexCeY{z}{\omega} &\faJHPjTb \GyFczGjW{0} + e^{\WHYFCIFc \omega} \GyFczGjW{1}
\end{align*}
\end{defi}
Whenever we talk about quantum polynomial-time algorithms
or quantum verifiers, we always mean polynomial-time uniformly
generated quantum circuits consisting of some
universal set of gates.
For example, the gates \gctqElrH, \lDGygnkY, and \sYcexCeY{z}{\frac{\pi}{4}}
form a universal set, and there are many other sets
that are universal.
Usually it doesn't matter which set we choose when we
define quantum verifiers and classes like \SHvwttaw{} and \JWevirQe,
because it is known that each universal set can approximate
any other set with exponential precision.
However, in the following we will use error parameters
that are smaller then this, hence the power of those
classes may change with respect to what set of gates
the verifier is using.
In this paper we only assume that the verifier
can perform or perfectly simulate the \gctqElrH{} and
the \lDGygnkY{} gate, besides being able to perform any
polynomial-time classical computation.
This assumption is enough for all our results, so
we won't bother about the gate set in the rest of the paper.
\begin{defi}[\cite{Kobayashi2003,Aaronson2009}]
\label{TJOqWPoY}
For functions \FyOFSCmf{\ell, k}{\CYLQkbAO}{\VrXQeMta},
\FyOFSCmf{c, s}{\CYLQkbAO}{ \left( 0, 1 \right] }
a language $L$ is in \KOzxJkVZ{\ell}{k}{c}{s} if there
exist a quantum verifier \wQoFYwHz{V} \ygMPbVVP{} for all
$ n \in \CYLQkbAO $ and inputs $ x \in \GGfSpoIn{0,1}^n $,
\HxLQfJAj{\wQoFYwHz{V}}{x} is a quantum circuit generated in
\wrDoIJTw{n}-time and the following holds.
\begin{description}
\item[Completeness:]
If $x \in L$ then there exist quantum proofs
\raIfJZxu{\phi_1}, \ldots, \raIfJZxu{\phi_{\HxLQfJAj{k}{n}}},
where for all $i$, \raIfJZxu{\phi_i} is made up of
at most \HxLQfJAj{\ell}{n} qubits, and the acceptance
probability of \HxLQfJAj{\wQoFYwHz{V}}{x} on input
$ \raIfJZxu{\phi_1} \NUxdODWZ \cdots \NUxdODWZ
\raIfJZxu{\phi_{\HxLQfJAj{k}{n}}} $ is $ \geq \HxLQfJAj{c}{n} $.
\item[Soundness:]
If $x \notin L$ then for all states
\raIfJZxu{\phi_1}, \ldots, \raIfJZxu{\phi_{\HxLQfJAj{k}{n}}},
where for all $i$, \raIfJZxu{\phi_i} is made up of
at most \HxLQfJAj{\ell}{n} qubits, \HxLQfJAj{\wQoFYwHz{V}}{x}
accepts with probability $ \leq s $,
given $ \raIfJZxu{\phi_1} \NUxdODWZ \cdots \NUxdODWZ
\raIfJZxu{\phi_{\HxLQfJAj{k}{n}}} $ as its input.
\end{description}
If the class is denoted by \KOzxJkVZ{\ell}{k}{c}{<s}, then
the probability bound in the soundness case is
$ < s $ instead of $ \leq s $.
\end{defi}
\begin{rem}
If we just give one parameter to \SHvwttaw, then
it indicates the number of provers.
So the notation \lrcGeBqo{k} is defined as
$ \lrcGeBqo{k} \faJHPjTb \KOzxJkVZ{\wrDoIJTw{n}} {k} {\frac{2}{3}} {\frac{1}{3}} $.
With our notation the class \SHvwttaw{} is defined as
$ \SHvwttaw \faJHPjTb \lrcGeBqo{1} $.
\end{rem}
\begin{defi}[\cite{Brandao2008,Aaronson2009}]
The class \fiWpgtYc{\ell}{k}{c}{s} is defined almost the same
way as \KOzxJkVZ{\ell}{k}{c}{s} in \MOVPSHbF{TJOqWPoY},
except that the verifier \wQoFYwHz{V} is not an
arbitrary \ipeoICuS-time quantum computation.
The restriction we put on the verifier is the following.
\wQoFYwHz{V} upon seeing $x$ performs a classical
randomized \ipeoICuS-time computation and produces
circuits for measurements $M_1$, \ldots, $M_{\HxLQfJAj{k}{n}}$,
where each $M_i$ is a POVM on at most \HxLQfJAj{\ell}{n} qubits.
Then for all $i$, the verifier measures
\raIfJZxu{\phi_i} with $M_i$, and obtains outcome $m_i$.
After all measurements were performed, \wQoFYwHz{V}
runs a classical computation on inputs
$m_1$, \ldots, $m_{\HxLQfJAj{k}{n}}$, and
decides whether to accept or reject.
\end{defi}
Note that in the above definition the verifier
has to measure each proofs separately.
Moreover, non of the circuits of the measurements can depend
on the outcome of any previous measurement.
Chen and Drucker \cite{Chen2010} defined \cdToKcBp{}
in a slightly different way by allowing the verifier
to do quantum computations before and after
the measurements.
Our result also holds if we take their definition.
The reason we chose the above definition is
because we will prove a lower-bound for our
\cdToKcBp{} class, and so with the more restricted
definition our result is slightly stronger.
\begin{defi}
The class \vtZmVJmM{c}{s} is defined analogously to
\SHvwttaw{} of \MOVPSHbF{TJOqWPoY} with the difference
that the proof must be a classical string.
We will always take this string to be polynomial-length.
Since the string is classical, it doesn't make sense
to have several proofs, so we drop the parameters of
proof lengths and number of proofs.
As in the previous definitions, $c$ is the completeness
and $s$ is the soundness parameter.
\end{defi}
\begin{defi}[\cite{Galperin1983}]
Let \HxLQfJAj{G}{V,E} be an undirected graph where
$V = \HelzWRDf{m}$ and $ m \leq 2^n $ for some $n$.
We define \oVrArWkh{} to be a \emph{small circuit representation}
of $G$ if the following hold:
\begin{itemize}
\item \oVrArWkh{} is a circuit containing AND, OR and NOT gates.
\item \oVrArWkh{} has two inputs of $n$ bits each.
\item \oVrArWkh{} has \wrDoIJTw{n} gates.
\item The output of \oVrArWkh{} is given by
\[
\HxLQfJAj{\oVrArWkh}{u,v} =
\begin{cases}
00 & \text{if } u \notin V \text{ or } v \notin V
\text{ or } u \geq v \text{,} \\
10 & \text{if } u < v \text{ and } \DnoYoBfx{u, v} \notin E \text{,} \\
11 & \text{if } u < v \text{ and } \DnoYoBfx{u, v} \in E \text{.}
\end{cases}
\]
\end{itemize}
\end{defi}
\begin{defi}
Let the decision problem \HNGSAaAW{} be the set of small circuit representations
of graphs that are 3-colorable.
\end{defi}
\begin{thm}[\cite{Papadimitriou1986}]
\HNGSAaAW{} is \vxJDyEiW{}-complete.
\end{thm}
In the following discussion we will use the
`SWAP-test' of \cite{Barenco1997,Buhrman2001},
and the following property of this test.
\begin{thm}[\cite{Buhrman2001}]
When the SWAP-test is applied to two states
\raIfJZxu{\varphi} and \raIfJZxu{\psi} (with the same dimension),
it accepts with probability
$ \frac{1}{2} \DnoYoBfx{1 + \AfBMaxZI{\NChmtXsR{\varphi}{\psi}}^2} $.
\label{QzwgCixS}
\end{thm}
Note that in order to perform the SWAP-test we need
two Hadamard (\lDGygnkY) gates, \gctqElrH{} gates with the
amount linear in the number of qubits
\raIfJZxu{\varphi} and \raIfJZxu{\psi} are stored on,
and we need to measure a qubit in the standard basis.
\begin{defi}
We define the state \raIfJZxu{u_m} as the uniform superposition
of the standard basis states, \WpedDmmP{}
\[ \raIfJZxu{u_m} \faJHPjTb \frac{1}{\sqrt{m}} \sum_{i=0}^{m-1} \raIfJZxu{i} . \]
We also define the projective measurement
that projects onto this state, or more formally
the measurement
$ \GGfSpoIn{ \nVzyhtCm{P}_0, \nVzyhtCm{P}_1 } $ where
\[ \nVzyhtCm{P}_0 \faJHPjTb \GyFczGjW{u_m} \text{,} \qquad
\nVzyhtCm{P}_1 \faJHPjTb \YMElKjII - \GyFczGjW{u_m} \text{.} \]
We say that $\nVzyhtCm{P}_0$ (and $\nVzyhtCm{P}_1$)
corresponds to outcome $0$ (and $1$).
\label{NlcnHZqK}
\end{defi}
Note that the above measurement can be performed
using \NpwuHusD{\log m} Hadamard gates and single-qubit
measurements.
\section{\texorpdfstring{\lrcGeBqo{k} with Small Gap Equals \vxJDyEiW}
{QMA[k] with Small Gap Equals NEXP}}
\label{jvxzezwT}
This section proves \MJqbMYrm{NmPOZGKI}, \WpedDmmP{}
we show that \lrcGeBqo{k} equals \vxJDyEiW{} if $k$ is at least
$2$ and at most \wrDoIJTw{n}, and the completeness-soundness
gap is bounded away by an inverse-exponential or doubly exponential
function of $n$.
\par
The proof of this theorem is divided into \dWQdKYxc{ULcFIvTl}
and \MJqbMYrm{QoOYZYVy},
according to the two directions of the containment.
As mentioned in the Introduction,
only \MJqbMYrm{QoOYZYVy} is actually new.
\begin{lem}
$ \displaystyle \bigcup_{\substack{0 < s < c \leq 1, \\
c-s \geq 2^{ -2^{\ipeoICuS} }}}
\KOzxJkVZ{\ipeoICuS}{\ipeoICuS}{c}{s}
\subseteq \vxJDyEiW $,
where \HxLQfJAj{c}{n} and \HxLQfJAj{s}{n} can be calculated in time
at most exponential in $n$ on a classical computer.
\label{ULcFIvTl}
\end{lem}
\begin{proof}[Proof sketch]
Let $ L \in \KOzxJkVZ{\ipeoICuS}{\ipeoICuS}{c}{s} $ with some
$c$ and $s$ satisfying the conditions in the lemma.
The proofs in the \SHvwttaw{} proof system are \ipeoICuS-many quantum states on
\ipeoICuS-many qubits, which are vectors in the complex euclidean
space with exponential dimension.
These vectors can be described up to an exponential number of
bits of accuracy by a classical proof of exponential length.
Given this proof to a \vxJDyEiW{} machine, it can calculate the
acceptance probability of the \SHvwttaw{} verifier to an exponential
number of bits of accuracy;
and it can decide whether this probability is more
then $c$ or less then $s$.
This means that $ L \in \vxJDyEiW $.
\end{proof}
The other direction of the containment is formulated
by the following theorem.
\begin{thm}
$ \displaystyle \vxJDyEiW \subseteq
\KOzxJkVZ{\ipeoICuS}{2}{1}{1 - \uaivIgOh{4^{-n}}} $.
\label{QoOYZYVy}
\end{thm}
We will prove this result through several lemmas.
This will prove \MJqbMYrm{NmPOZGKI}
as well.
\begin{proof}[Proof of \MJqbMYrm{NmPOZGKI}]
The theorem immediately follows from \dWQdKYxc{ULcFIvTl},
\MJqbMYrm{QoOYZYVy} and the observation that
$ \KOzxJkVZ{\ipeoICuS}{2}{1}{1 - \uaivIgOh{4^{-n}}} \subseteq
\bigcup_{\substack{0 < s < c \leq 1, \\
c-s \geq 2^{ -2^{\ipeoICuS} }}}
\KOzxJkVZ{\ipeoICuS}{\ipeoICuS}{c}{s} $.
\end{proof}
In order to prove \MJqbMYrm{QoOYZYVy} we
construct a \SHvwttaw{} verifier \wQoFYwHz{V} for
the \vxJDyEiW-complete language \HNGSAaAW.
We use the protocol of Blier and Tapp \cite{Blier2009},
with similar argument as theirs.
\par
Let the input to \HNGSAaAW{} be denoted by \oVrArWkh{} and it's
length by $ n = \AfBMaxZI{\oVrArWkh} $.
Let \wQoFYwHz{V} get its two unentangled proofs in registers
$ \VSdxZJNO{R}_1 $ and $ \VSdxZJNO{R}_2 $.
Both $ \VSdxZJNO{R}_i $'s have two parts,
$ \VSdxZJNO{R}_i = \VSdxZJNO{N}_i \VSdxZJNO{C}_i $,
where $\VSdxZJNO{N}_i$ is the `node' part and
$\VSdxZJNO{C}_i$ is the `color' part.
$\VSdxZJNO{N}_1$ and $\VSdxZJNO{N}_2$ have associated Hilbert space
$ \VWqEPttf^{2^n} $, while $\VSdxZJNO{C}_1$ and $\VSdxZJNO{C}_2$
have associated space $ \VWqEPttf^3 $.
The procedure \wQoFYwHz{V} performs is described in
\FxucuSoM{jEVcePOJ}.
\par
\begin{algorithm}[h t b]
\caption{Description of verifier \wQoFYwHz{V} in the proof of
\MJqbMYrm{QoOYZYVy}.}
\label{jEVcePOJ}
\begin{algorithmic}[1]
\ZNbYlnmP classical circuit \oVrArWkh, quantum registers
$\VSdxZJNO{R}_1$ and $\VSdxZJNO{R}_2$,
where the state of $\VSdxZJNO{R}_1$ and
$\VSdxZJNO{R}_2$ is separable
\tkOoHzex accept or reject
\STATE With probability $1/3$ do
the \textbf{Equality Test} (line~\ref{yeMuKErF}),
the \textbf{Consistency Test} (line~\ref{YKlinzMJ}), or
the \textbf{Uniformity Test} (line~\ref{DOkcHsyK}).
\STATE \textbf{Equality Test.} Perform the SWAP-test on
$ \VSdxZJNO{R}_1 $ and $ \VSdxZJNO{R}_2 $.
\label{yeMuKErF}
\IF{the SWAP-test fails}
\RETURN reject
\COMMENT{The two registers are not equal.}
\ELSE
\RETURN accept
\ENDIF
\STATE \textbf{Consistency Test.}
\label{YKlinzMJ}
Measure $ \VSdxZJNO{N}_1 $, $ \VSdxZJNO{C}_1 $,
$\VSdxZJNO{N}_2 $ and $ \VSdxZJNO{C}_2 $
in the computational basis and get
$v_1$, $c_1$, $v_2$ and $c_2$.
\IF{($ v_1 = v_2 $) \AND ($ c_1 \neq c_2 $)} \label{DkvnMEbW}
\RETURN reject
\COMMENT{The same vertex has two colors.}
\ELSIF{($ \HxLQfJAj{\wQoFYwHz{C}}{v_1, v_2} = 11 $)
\COMMENT{Assume that $ v_1 < v_2 $ otherwise swap them.}
\AND ($ c_1 = c_2 $)}
\RETURN reject
\COMMENT{Adjacent vertices have same color.}
\ELSE
\RETURN accept
\ENDIF
\STATE \textbf{Uniformity Test.}
Measure $ \VSdxZJNO{N}_1 $ and $ \VSdxZJNO{C}_1 $ separately
according to the measurement of
\MOVPSHbF{NlcnHZqK}.
\label{DOkcHsyK}
\IF{(the outcome on $ \VSdxZJNO{C}_1 $ is $0$)
\AND (the outcome on $ \VSdxZJNO{N}_1 $ is $1$)}
\label{GTSlDeLX}
\RETURN reject
\COMMENT{Not all nodes are present.}
\ELSE
\RETURN accept
\ENDIF
\end{algorithmic}
\end{algorithm}
Note that \wQoFYwHz{V} runs in \wrDoIJTw{n}-time,
because the SWAP-test, evaluating the circuit \oVrArWkh,
and performing the measurement of
\MOVPSHbF{NlcnHZqK} for
$ m=2^n $ all can be performed in polynomial time.
The following lemma,
which is essentially the same as Theorem~2.4 of \cite{Blier2009},
proves completeness for \wQoFYwHz{V}.
\begin{lem}[Completeness]
If $ \oVrArWkh \in \HNGSAaAW $ then there exist a pair of proofs,
with which \wQoFYwHz{V} will
accept with probability $1$.
\label{ueMCbSRk}
\end{lem}
\begin{proof}
For $ i \in \HelzWRDf{m} $ let $ \HxLQfJAj{c}{i} \in \GGfSpoIn{0,1,2} $
be a valid coloring of the graph $G$,
where $m$ is the number of nodes.
For $ i \in \GGfSpoIn{m, \ldots, 2^n-1} $ let $ \HxLQfJAj{c}{i} = 0 $.
Let the state of both $ \VSdxZJNO{R}_1 $ and $ \VSdxZJNO{R}_2 $ be
\begin{align}
\raIfJZxu{\phi} = \frac{1}{\sqrt{2^n}}
\sum_{i=0}^{2^n - 1} \raIfJZxu{i} \raIfJZxu{\HxLQfJAj{c}{i}} \text{,}
\label{BGuJsFqG}
\end{align}
where \raIfJZxu{i} is on the node register (\VSdxZJNO{N})
and \raIfJZxu{\HxLQfJAj{c}{i}} is on the color register (\VSdxZJNO{C}).
From \MJqbMYrm{QzwgCixS}
it follows that the Equality Test succeeds with probability $1$.
Since $ c $ is a valid 3-coloring, the Consistency Test
succeeds with probability $1$ as well.
To see the same for the Uniformity Test,
let us calculate the resulting state
after measuring $0$ on $ \VSdxZJNO{C}_1 $ in
line~\ref{GTSlDeLX} of
\FxucuSoM{jEVcePOJ}.
Up to some normalization factor, the state is
\[ \DnoYoBfx{\YMElKjII \NUxdODWZ \GyFczGjW{u_3}} \raIfJZxu{\phi} =
\frac{1}{3 \sqrt{2^n}} \DnoYoBfx{\sum_{i=0}^{2^n - 1} \raIfJZxu{i}}
\NUxdODWZ \DnoYoBfx{\sum_{k=0}^{2} \raIfJZxu{k}} . \]
This means that the state of $ \VSdxZJNO{N}_1 $
is \raIfJZxu{u_{2^n}}, so the Uniformity Test always succeeds.
\end{proof}
The proof of soundness is presented in \IhApwfvN{WEFVXHAK}
on page~\pageref{WEFVXHAK}, because it
closely follows the analysis of Blier and Tapp \cite{Blier2009}
and of Chiesa and Forbes \cite{Chiesa2011},
although with different parameters.
\section{\texorpdfstring{\pUeVeClC{\uaivIgOh{n}} with Small Gap Equals \vxJDyEiW}
{BellQMA[Omega(n)] with Small Gap Equals NEXP}}
\label{BXwCHKEG}
In this section we prove \MJqbMYrm{nKVfzRHR},
\WpedDmmP{} we show that \SHvwttaw{} with exponentially
small gap still equals to \vxJDyEiW{} if we restrict
the verifier to only be able to perform Bell-measurements.
However, we will need at least \uaivIgOh{n} proofs.
We essentially use the algorithm of Chen
and Drucker \cite{Chen2010} on the succinct version of graph
3-coloring (\HNGSAaAW).
We also use one of their lemmas, but
our proof will be simpler then theirs,
because we don't aim for constant gap.
We don't use the PCP theorem either.
Note that in the previous section we already
argued about the \vxJDyEiW{} upper-bound on the
\SHvwttaw{} classes.
The same argument applies here too.
It's also easy to see that restricting the verifier
can only make the power of the proof system weaker.
So the only statement left to prove, in order to prove
\MJqbMYrm{nKVfzRHR}, is the following.
\begin{thm}
$ \vxJDyEiW \subseteq \fiWpgtYc{\wrDoIJTw{n}}{\uaivIgOh{n}}{c}{s} $,
for some $c$ and $s$
with $ \HxLQfJAj{c}{n} - \HxLQfJAj{s}{n} = \uaivIgOh{4^{-n}} $.
\label{okAAwxmt}
\end{thm}
To prove \MJqbMYrm{okAAwxmt},
we construct a \cdToKcBp{} verifier \wQoFYwHz{V} for
\HNGSAaAW{} which will show that
$ \HNGSAaAW{} \in \fiWpgtYc{\wrDoIJTw{n}}{\uaivIgOh{n}}{c}{s} $.
Since \HNGSAaAW{} is \vxJDyEiW-complete,
\MJqbMYrm{okAAwxmt} will follow.
Just as in the previous section,
let the input to \HNGSAaAW{} be denoted by \oVrArWkh{} and it's
length by $ n = \AfBMaxZI{\oVrArWkh} $.
Verifier \wQoFYwHz{V} will receive $k$ quantum proofs
in registers
$\VSdxZJNO{N}_1, \VSdxZJNO{C}_1, \ldots, \VSdxZJNO{N}_k, \VSdxZJNO{C}_k$,
where for each $ i \in \iteYiYHH{k} $,
the state of $\VSdxZJNO{N}_i \VSdxZJNO{C}_i$ is
separable from the rest of the registers.
The registers $\VSdxZJNO{N}_i$ have associated space
$ \VWqEPttf^{2^n} $ and registers $\VSdxZJNO{C}_i$
have associated space $ \VWqEPttf^3 $, similarly as in
the previous section.
The behavior of \wQoFYwHz{V} is described
in \FxucuSoM{NuAJSVLc}.
\begin{algorithm}[h t b]
\caption{Description of verifier \wQoFYwHz{V} in the proof of
\MJqbMYrm{okAAwxmt}.}
\label{NuAJSVLc}
\begin{algorithmic}[1]
\ZNbYlnmP classical circuit \oVrArWkh, quantum registers
$\VSdxZJNO{N}_1, \VSdxZJNO{C}_1, \ldots,
\VSdxZJNO{N}_k, \VSdxZJNO{C}_k$,
where $ \forall i \in \iteYiYHH{k} $,
the state of $\VSdxZJNO{N}_i \VSdxZJNO{C}_i$ is
separable from the rest of the registers
\tkOoHzex accept or reject
\STATE With probability $ \frac{1}{2} $ do the
\textbf{Consistency Test} (line~\ref{NgopDguv})
or the \textbf{Uniformity Test}
(line~\ref{qMwgVWBk}).
\label{GJhSCbBL}
\STATE \textbf{Consistency Test.}
\label{NgopDguv}
\FORALL{$ i \in \iteYiYHH{k} $}
\STATE Measure $\VSdxZJNO{N}_i$ and $\VSdxZJNO{C}_i$
in the computational basis and get $v_i$ and $c_i$.
\label{qAGbukNT}
\ENDFOR
\FORALL{$ 1 \leq i < j \leq k $}
\IF{($ v_i = v_j $) \AND ($ c_i \neq c_j $)}
\RETURN reject
\COMMENT{The same vertex has two colors.}
\label{FMIMsMqc}
\ELSIF{($ \HxLQfJAj{\wQoFYwHz{\oVrArWkh}}{v_i, v_j} = 11 $)
\AND ($ c_i = c_j $)}
\RETURN reject
\COMMENT{Adjacent vertices have same color.}
\label{eIpHlpwC}
\ENDIF
\ENDFOR
\RETURN accept
\STATE \textbf{Uniformity Test.}
\label{qMwgVWBk}
\FORALL{$ i \in \iteYiYHH{k} $}
\STATE Measure $ \VSdxZJNO{C}_i $ with the measurement
of \MOVPSHbF{NlcnHZqK}
and denote the outcome by $x_i$.
\label{ZtoStVJz}
\STATE Measure $ \VSdxZJNO{N}_i $ with the measurement
of \MOVPSHbF{NlcnHZqK}
and denote the outcome by $y_i$.
\label{fEiSBHoN}
\ENDFOR
\STATE Let $ Z \faJHPjTb \lUlEWeUc{i}{x_i = 0} $.
\IF{$ \AfBMaxZI{Z} < k/6 $}
\RETURN reject
\label{czNyFhpM}
\ENDIF
\FORALL{$ i \in Z $}
\IF{$ y_i = 1 $}
\RETURN reject
\COMMENT{Not all nodes are present.}
\ENDIF
\ENDFOR
\RETURN accept
\end{algorithmic}
\end{algorithm}
We split the proof of \MJqbMYrm{okAAwxmt}
into \dWQdKYxc{iQoQORKN},
which proves completeness for \wQoFYwHz{V} and
\dWQdKYxc{gFZRHSOy}, which proves its soundness.
\begin{lem}[Completeness]
If $ \oVrArWkh \in \HNGSAaAW $ then there exist quantum states
on registers $\VSdxZJNO{N}_1, \VSdxZJNO{C}_1, \ldots,
\VSdxZJNO{N}_k, \VSdxZJNO{C}_k$, such that if they are input to
\wQoFYwHz{V}, defined by \FxucuSoM{NuAJSVLc},
then \wQoFYwHz{V} will accept with probability at least
$ 1 - 2^{- \frac{k}{40}} $.
\label{iQoQORKN}
\end{lem}
\begin{proof}
For all $ i \in \iteYiYHH{k} $, let the state of
$\VSdxZJNO{N}_i \VSdxZJNO{C}_i$ be \raIfJZxu{\phi},
where \raIfJZxu{\phi} is defined by equation~\eqref{BGuJsFqG}
on page~\pageref{BGuJsFqG}.
For exactly the same reason as in the proof of
\dWQdKYxc{ueMCbSRk} the Consistency Test will
succeed with probability $1$.
As for the Uniformity Test, note that for all
$i \in Z$, the measurement of $\VSdxZJNO{N}_i$
in line~\ref{fEiSBHoN}
of \FxucuSoM{NuAJSVLc} yields $1$ with
probability $1$.
The argument for this is also in the proof of
\dWQdKYxc{ueMCbSRk}.
\par
This means that given the above input, the only place where
\FxucuSoM{NuAJSVLc} may reject is at
line~\ref{czNyFhpM}, \WpedDmmP{}
when $ \AfBMaxZI{Z} < \frac{k}{6} $.
So in the following we only need to upper-bound
this probability.
We do it similarly to the proof of Lemma~1 of \cite{Chen2010}.
By direct calculation the probability that
$x_i = 0$ in line~\ref{ZtoStVJz} is
\[ \TdSeQgKh{x_i = 0} =
\fHAwcqZo{\phi} \DnoYoBfx{\YMElKjII \NUxdODWZ \GyFczGjW{u_3}} \raIfJZxu{\phi} =
\frac{1}{3} . \]
This means that $ \YcVNkVQD{\AfBMaxZI{Z}} = \frac{k}{3} $.
Since the $x_i$'s are independent, we can use the
\fKqDCezI{} and get that
\[ \TdSeQgKh{\AfBMaxZI{Z} < \frac{k}{6}} <
e^{- \frac{k}{48}} <
2^{- \frac{k}{40}} . \]
This finishes the proof of the lemma.
\end{proof}
\newcommand{\vXoALQtP}[1]{\ensuremath{ \alpha_{v}^{\DnoYoBfx{#1}} }}
\newcommand{\yJgNrafg}[1]{\ensuremath{ \beta_{v,j}^{\DnoYoBfx{#1}} }}
\newcommand{\zZTHRpCM}{\ensuremath{ \gamma_{v}^{\DnoYoBfx{i}} }}
We are left to prove soundness for \wQoFYwHz{V}.
From now on let's denote the quantum input to
\FxucuSoM{NuAJSVLc} by \raIfJZxu{\varphi_1}, \ldots,
\raIfJZxu{\varphi_k}.
For each $ i \in \iteYiYHH{k} $, we write
\[ \raIfJZxu{\varphi_i} = \sum_{v=0}^{2^n-1} \vXoALQtP{i} \raIfJZxu{v}
\sum_{j=0}^2 \yJgNrafg{i} \raIfJZxu{j} \text{,} \]
where \raIfJZxu{v} is a state on $\VSdxZJNO{N}_i$,
\raIfJZxu{j} is a state on $\VSdxZJNO{C}_i$, furthermore
$ \sum_{v=0}^{2^n-1} \AfBMaxZI{\vXoALQtP{i}}^2 = 1 $ for each $i$, and
$ \sum_{j=0}^2 \AfBMaxZI{\yJgNrafg{i}}^2 = 1 $ for each $i$ and $v$.
Similarly to the notation in \cite{Chen2010} let
\[ Z' \faJHPjTb \lUlEWeUc{i}{ \TdSeQgKh{x_i = 0} \geq \frac{1}{12} } . \]
\par
We need a lemma from \cite{Chen2010} which we will
state and use with a bit different parameters.
Intuitively the lemma says that in order to avoid rejection
in line~\ref{czNyFhpM}, we must
have a constant fraction of registers for which,
with at least a constant probability, the outcome
of the measurement in
line~\ref{ZtoStVJz} is $0$.
\begin{lem}[Lemma~2 of \cite{Chen2010}]
If $ \AfBMaxZI{Z'} \leq k/6 $ and if in
line~\ref{GJhSCbBL} of
\FxucuSoM{NuAJSVLc} the Uniformity
Test is chosen, then the test will reject in
line~\ref{czNyFhpM} with
probability \uaivIgOh{1}.
\label{vzMhnkxb}
\end{lem}
We want all nodes to appear with sufficiently
big amplitude in \raIfJZxu{\varphi_i}, for each $i \in Z'$.
This is formalized by the following lemma.
\begin{lem}
Suppose that the Uniformity Test rejects with probability
at most $ 200^{-1} \cdot 4^{-n} $.
Then $ \forall i \in Z' $ and $ \forall v \in \HelzWRDf{2^n} $
it holds that
\[ \AfBMaxZI{\vXoALQtP{i}}^2 > \frac{1}{24 \cdot 2^n} . \]
\label{UDjzaeqm}
\end{lem}
\begin{proof}
Let's pick an $i \in Z'$, and consider the state
\raIfJZxu{\varphi_i} on register $\VSdxZJNO{N}_i \VSdxZJNO{C}_i$.
Suppose that we measured $0$ on $\VSdxZJNO{C}_i$ with
the measurement of \MOVPSHbF{NlcnHZqK},
and denote the resulting state on $\VSdxZJNO{N}_i$ by
\[ \raIfJZxu{\xi_i} = \sum_{v=0}^{2^n-1} \zZTHRpCM \raIfJZxu{v} . \]
Note that since $i \in Z'$, this outcome happens with
probability $ \geq \frac{1}{12} $.
Assume towards contradiction that $ \exists v $ \ygMPbVVP{}
$ \AfBMaxZI{\zZTHRpCM}^2 < \frac{1}{2 \cdot 2^n} $.
Using \dWQdKYxc{tAaagiyx}
from page~\pageref{tAaagiyx},
we get that when we measure \raIfJZxu{\xi_i} with the measurement
of \MOVPSHbF{NlcnHZqK},
we get outcome $1$ with probability at least
$ \frac{1}{16 \cdot 4^{n}} $.
But this means that the Uniformity Test rejects with
probability at least
$ \frac{1}{12 \cdot 16 \cdot 4^{n}} >
\frac{1}{200 \cdot 4^{n}} $.
This contradicts to the statement of the lemma,
so it must be that
$ \AfBMaxZI{\zZTHRpCM}^2 \geq \frac{1}{2 \cdot 2^n} $
for all $v$.
\dWQdKYxc{KHySvfFa} implies that
for all $v$,
\[ \AfBMaxZI{\vXoALQtP{i}}^2 \geq
\frac{1}{12 \cdot 2 \cdot 2^n} . \qedhere \]
\end{proof}
We are now ready to prove soundness for \wQoFYwHz{V}.
\begin{lem}[Soundness]
If $ \oVrArWkh \notin \HNGSAaAW $ then \wQoFYwHz{V} of
\FxucuSoM{NuAJSVLc} will reject with
probability at least $ 12000^{-1} \cdot 4^{-n} $.
\label{gFZRHSOy}
\end{lem}
\begin{proof}
Suppose that the Uniformity Test rejects with probability at most
$ \frac{1}{200 \cdot 4^{n}} $,
as otherwise we are done.
From \dWQdKYxc{vzMhnkxb}, $ \AfBMaxZI{Z'} > \frac{k}{6} $.
Since $ \frac{k}{6} = \uaivIgOh{n} $, we can always take
$ k \geq 12 $ so we have $ \AfBMaxZI{Z'} > 2 $.
Let's pick two elements $ q,r \in Z' $.
We define two colorings $c_1$ and $c_2$ the
following way,
\begin{align*}
\HxLQfJAj{c_1}{v} &\faJHPjTb \argmax_j \AfBMaxZI{\yJgNrafg{q}}
\intertext{and similarly}
\HxLQfJAj{c_2}{v} &\faJHPjTb \argmax_j \AfBMaxZI{\yJgNrafg{r}} \text{,}
\end{align*}
for all $ v \in \HelzWRDf{2^n} $.
If the maximum is not well defined then just choose
an arbitrary $j$ for which \AfBMaxZI{\yJgNrafg{.}} is maximal.
From \dWQdKYxc{UDjzaeqm}
the probability that we get \DnoYoBfx{v, \HxLQfJAj{c_1}{v}}
when we measure \raIfJZxu{\varphi_q} in the standard basis is
at least $ \AfBMaxZI{\vXoALQtP{q}}^2 \cdot \frac{1}{3} >
\frac{1}{72 \cdot 2^n} $, for all $v$,
and the same lower-bound is true for getting
\DnoYoBfx{v, \HxLQfJAj{c_2}{v}}, when measuring \raIfJZxu{\varphi_r}.
We split the rest of the proof into two cases.
\begin{itemize}
\item Suppose that the two colorings are different, \WpedDmmP{}
$ \exists v $ \ygMPbVVP{} $ \HxLQfJAj{c_1}{v} \neq \HxLQfJAj{c_2}{v} $.
In this case in line~\ref{qAGbukNT}
we get \DnoYoBfx{v, \HxLQfJAj{c_1}{v}} when measuring
$\VSdxZJNO{N}_q \VSdxZJNO{C}_q$ and
\DnoYoBfx{v, \HxLQfJAj{c_2}{v}} when measuring $\VSdxZJNO{N}_r \VSdxZJNO{C}_r$
with probability at least
$ \DnoYoBfx{\frac{1}{72 \cdot 2^n}}^2 >
\frac{1}{6000 \cdot 4^{n}} $.
It means that with at least the above probability
the Consistency Test will reject in
line~\ref{FMIMsMqc}.
\item Suppose that the two colorings are the same, \WpedDmmP{}
$ \forall v : \HxLQfJAj{c_1}{v} = \HxLQfJAj{c_2}{v} $.
Since $G$ is not 3-colorable,
$ \exists v_1, v_2 \in \HelzWRDf{2^n} $ \ygMPbVVP{}
\DnoYoBfx{v_1, v_2} is an edge in $G$ and
$ \HxLQfJAj{c_1}{v_1} = \HxLQfJAj{c_1}{v_2} $, or equivalently
$ \HxLQfJAj{\oVrArWkh}{v_1, v_2} = 11 $.
Similarly as above, with probability at least
$ \frac{1}{6000 \cdot 4^{n}} $,
we get \DnoYoBfx{v_1, \HxLQfJAj{c_1}{v_1}} when measuring
$\VSdxZJNO{N}_q \VSdxZJNO{C}_q$ and \DnoYoBfx{v_2, \HxLQfJAj{c_1}{v_2}}
when measuring $\VSdxZJNO{N}_r \VSdxZJNO{C}_r$
in line~\ref{qAGbukNT}
of the algorithm.
In this case the Consistency Test will reject with
at least the above probability in
line~\ref{eIpHlpwC}.
\end{itemize}
Since in both cases the Consistency Test rejects with
probability at least $ \frac{1}{6000 \cdot 4^{n}} $,
and the test is chosen with probability $\frac{1}{2}$,
the lemma follows.
\end{proof}
\begin{proof}[Proof of \MJqbMYrm{okAAwxmt}]
Note that \FxucuSoM{NuAJSVLc} runs in
polynomial time.
Furthermore, for both the Consistency and
the Uniformity Test, the algorithm starts with
measuring all the quantum registers according to a
fixed measurement.
So \wQoFYwHz{V} is a proper \cdToKcBp{} verifier.
\dWQdKYxc{iQoQORKN} shows that the
completeness of the protocol is
$ c > 1 - 2^{- \frac{k}{40}} $, while
\dWQdKYxc{gFZRHSOy} shows that the
soundness is $ s < 1 - 12000^{-1} \cdot 2^{-2n} $.
If $ k \geq 120 n $ then $ c - s = \uaivIgOh{2^{-2n}} $
so the theorem follows.
\end{proof}
\section{Conclusions and Open Problems}
\label{HujYVLzw}
In this section we discuss some of the consequences
of the previous results, \WpedDmmP{}
the consequences of \MJqbMYrm{NmPOZGKI} and
\ref{nKVfzRHR}.
We also raise some related open problems.
\subsection{Tightness of the Soundness Analyses}
One can observe that both the \lrcGeBqo{2} verifier
of \FxucuSoM{jEVcePOJ} and the \cdToKcBp{}
verifier of \FxucuSoM{NuAJSVLc}
have soundness parameter $ 1 - \uaivIgOh{4^{-n}} $
and gap $\uaivIgOh{4^{-n}}$.
(As shown by \dWQdKYxc{SaalgorQ}
and \dWQdKYxc{gFZRHSOy}.)
Note that this bound is tight up to a constant
factor in case of \FxucuSoM{jEVcePOJ} and
tight up to some low-order term in case of
\FxucuSoM{NuAJSVLc}.
The reason for this is the same as what was
observed in one of the remark in \cite{Chiesa2011}.
\par
The argument is briefly the following.
Suppose that $ \oVrArWkh \notin \HNGSAaAW $ and
that $G$ is such, that there exist a coloring
\ygMPbVVP{} only one pair of nodes are colored inconsistently.
If the prover gives states of the form defined by
equation~\eqref{BGuJsFqG} but using this
coloring, then the verifier won't notice this
in either of the Uniformity Tests nor in the
Equality Test.
The only place where the verifier can catch the prover
is in the Consistency Test when it checks the colors
of the nodes according to the constraints posed by the graph $G$.
The prover gets caught if the verifier gets
the inconsistently colored nodes in the measurement outputs.
This happens with probability \QUPOtzml{2^{-2n}}
in case of \FxucuSoM{jEVcePOJ} and
\QUPOtzml{n^2 2^{-2n}} in case of \FxucuSoM{NuAJSVLc}.
This means that it is possible to fool the verifier
of \FxucuSoM{jEVcePOJ} with probability
$ 1 - \QUPOtzml{4^{-n}} $ and to fool the verifier of
\FxucuSoM{NuAJSVLc} with probability
$ 1 - \CNzHFRry{4^{-n}} $.
\subsection{\texorpdfstring{Separation Between \SHvwttaw{} and \lrcGeBqo{2}
in the Small-Gap Setting}
{Separation Between QMA and QMA[2] in the Small-Gap Setting}}
\label{BhXaffZV}
As we said in the Introduction, it is a
big open problem whether \SHvwttaw{} is equal to
\lrcGeBqo{2} or not, and we mentioned some evidences
that suggest us that they are not equal.
Here we give another evidence, \WpedDmmP{}
we show that under plausible complexity-theoretic
assumptions \lrcGeBqo{2} is \emph{strictly} more
powerful than \SHvwttaw{} in the low-gap setting.
We elaborate on this in the following.
\par
\MJqbMYrm{NmPOZGKI} shows
that \lrcGeBqo{2} with exponentially or double-exponentially
small gap is exactly characterized by \vxJDyEiW.
So it is natural to ask, what is the power of
\SHvwttaw{} with the same gap, or what upper-bounds can
we give for it?
In a related paper, Ito \QgZuCZdC\ \cite{Ito2010}
showed that \XaOdbtsT{}s (or the class \JWevirQe) with double-exponentially
small gap are exactly characterized by \QZccOLoE.
Since \JWevirQe{} contains \SHvwttaw{} with the same gap,
we have a separation between \SHvwttaw{} and \lrcGeBqo{2}
in the setting where the gap is exponentially or
double-exponentially small, unless $ \QZccOLoE = \vxJDyEiW $.
Note that the result of Ito \QgZuCZdC\ is quite involved.
But if we are only interested in the upper-bound
on \SHvwttaw, then we can give a very simple argument for it,
which we state and prove in the following lemma.
\begin{lem}
$ \displaystyle \bigcup_{\substack{0 < s < c \leq 1, \\
c-s \geq 2^{ -2^{\ipeoICuS} }}}
\KOzxJkVZ{\ipeoICuS}{1}{c}{s}
\subseteq \QZccOLoE $,
where \HxLQfJAj{c}{n} and \HxLQfJAj{s}{n} can be calculated in time
at most exponential in $n$ on a classical computer.
\label{XYHrJdGo}
\end{lem}
\begin{proof}[Proof sketch]
Let $x$ be an input to \KOzxJkVZ{\ipeoICuS}{1}{c}{s} with $c$ and
$s$ having the above property.
The action of the verifier can be described by a
binary-valued measurement
\GGfSpoIn{\nVzyhtCm{P}_0^x, \nVzyhtCm{P}_1^x} on
the proof state, where $\nVzyhtCm{P}_1^x$ corresponds to
acceptance and $\nVzyhtCm{P}_0^x$ corresponds to
rejecting.
Note that the maximum acceptance probability of the
verifier is equal to the spectral norm of
$\nVzyhtCm{P}_1^x$.
(Or in other words the biggest eigenvalue of
$\nVzyhtCm{P}_1^x$.)
Since the proof is on \ipeoICuS-many qubits, the
dimension of $\nVzyhtCm{P}_1^x$ is exponential.
An \QZccOLoE-machine, knowing $x$, can approximate
$\nVzyhtCm{P}_1^x$ with up to an exponential amount
of digits of accuracy.
This is because the verifier is a uniform
quantum circuit of polynomial size.
Now the \QZccOLoE-machine can approximate the spectral
norm of $\nVzyhtCm{P}_1^x$ up to an exponential
amount of bits of accuracy.
\end{proof}
\subsubsection{One-Sided Error Case}
Note that the \vxJDyEiW{} characterization of the small-gap
\lrcGeBqo{2} proof system still holds if we restrict
the proof system to have one-sided error.
Moreover, it is not known whether \SHvwttaw{} can be
made to have one-sided error, so we can investigate
the relation between these classes as well.
Interestingly, it turns out that we can state
an even stronger separation in this case.
This is due to a result by Ito \QgZuCZdC\ \cite{Ito2010}.
\begin{thm}[Theorem~11 of \cite{Ito2010}]
$ \KOzxJkVZ{\ipeoICuS}{1}{1}{<1} \subseteq \LHqjxWJM $.
\end{thm}
This means that in the one-sided error case \lrcGeBqo{2}
with exponentially small gap is \emph{strictly}
more powerful then \SHvwttaw{} with even unbounded gap,
unless $ \LHqjxWJM = \vxJDyEiW $!
\subsection{Nonexistence of Disentanglers}
The discussions in \HVRLKXNY{BhXaffZV}
have an interesting consequence to the existence
question of certain operators called disentanglers.
They were defined by Aaronson \QgZuCZdC\ \cite{Aaronson2009},
as the following.
\begin{defi}[Definition~40 of \cite{Aaronson2009}]
Let us have a superoperator
\FyOFSCmf{\Phi} {\KKIedNYP{\VWqEPttf^N}}
{\KKIedNYP{\VWqEPttf^M \NUxdODWZ \VWqEPttf^M}}.\footnote{We
denote the set of all density operators on space
\zdLEuPyh{H} by \KKIedNYP{\zdLEuPyh{H}}.}
We say that $\Phi$ is an \DnoYoBfx{\varepsilon,
\delta}-disentangler if
\begin{itemize}
\item \HxLQfJAj{\Phi}{\rho} is $\varepsilon$-close
to a separable state for every $\rho$, and
\item for every separable state $\sigma$,
there exists a $\rho$ \ygMPbVVP{} \HxLQfJAj{\Phi}{\rho} is
$\delta$-close to $\sigma$.
\end{itemize}
\label{IYzbetOO}
\end{defi}
Note that if there exist a
\DnoYoBfx{\frac{1}{\wrDoIJTw{\log M}}, \frac{1}{\wrDoIJTw{\log M}}
}-disentangler with $ \log N = \wrDoIJTw{\log M} $,
and if that disentangler can be implemented in
\wrDoIJTw{\log M}-time, then $ \SHvwttaw = \lrcGeBqo{2} $.
So it is not believed that such a disentangler exist.
Towards proving this, Aaronson \QgZuCZdC\ showed that no \DnoYoBfx{0,0}-disentangler
exists for any finite $N$ and $M$.
The discussion in the previous section implies that
there exist no disentangler with approximation error
inverse of the square of the dimension.
More precisely we have the following corollary.
\begin{cor}
There exist a function $ \HxLQfJAj{f}{M} = \uaivIgOh{M^{-2}} $,
\ygMPbVVP{} there exist no \wrDoIJTw{\log M}-time implementable
\DnoYoBfx{\HxLQfJAj{f}{M}, \HxLQfJAj{f}{M}
}-disentangler with $ \log N = \wrDoIJTw{\log M} $,
unless $ \QZccOLoE = \vxJDyEiW $.
\end{cor}
\begin{proof}
Suppose that there exist such a disentangler $\Phi$,
for $ \HxLQfJAj{f}{M} = \kappa_1 \cdot M^{-2} $,
for some constant $ \kappa_1 $ to be specified later.
Let $ L \in \vxJDyEiW $.
\MJqbMYrm{QoOYZYVy} implies that
$ L \in \KOzxJkVZ{\ipeoICuS}{2}{1}{1 - \kappa_2 \cdot 4^{-n}} $,
for some constant $ \kappa_2 $ and where the
dimension of both proof states are $ 3 \cdot 2^n $.
Let \wQoFYwHz{V} be the corresponding verifier.
We show that $ L \in \KOzxJkVZ{\ipeoICuS}{1}{c}{s} $ with
$ c-s = \uaivIgOh{4^{-n}} $ by constructing a
verifier \wQoFYwHz{W} that uses only one proof.
By \dWQdKYxc{XYHrJdGo} it holds that
$ \KOzxJkVZ{\ipeoICuS}{1}{c}{s} \subseteq \QZccOLoE $
so we get that $ \QZccOLoE = \vxJDyEiW $.
\par
We are left to define verifier \wQoFYwHz{W}.
\wQoFYwHz{W} fist applies $\Phi$ on its quantum proof
then simulates \wQoFYwHz{V} on the output of $\Phi$
and outputs whatever \wQoFYwHz{V} outputs.
Note that $\Phi$ is \ipeoICuS-time implementable and
the size of the proof of \wQoFYwHz{W} is also
polynomial.
To see completeness for \wQoFYwHz{W}, note that
there exist a state $ \raIfJZxu{\psi} \NUxdODWZ \raIfJZxu{\psi} $
with which \wQoFYwHz{V} accepts with probability $1$.
From \MOVPSHbF{IYzbetOO} there exist a
$ \rho $ \ygMPbVVP{} \HxLQfJAj{\Phi}{\rho} is
\HxLQfJAj{f}{3 \cdot 2^n}-close to
$ \raIfJZxu{\psi} \NUxdODWZ \raIfJZxu{\psi} $.
Since $ \HxLQfJAj{f}{3 \cdot 2^n} =
\frac{\kappa_1}{9} \cdot 4^{-n} $,
the probability of acceptance of \wQoFYwHz{W} is
at least $ 1 - \frac{\kappa_1}{9} \cdot 4^{-n} $.
Similarly, for the soundness of \wQoFYwHz{W},
we have that for all separable states,
\wQoFYwHz{V} accepts with probability
at most $ 1 - \kappa_2 \cdot 4^{-n} $.
Again from \MOVPSHbF{IYzbetOO},
for all $\rho$, \HxLQfJAj{\Phi}{\rho} is
\HxLQfJAj{f}{3 \cdot 2^n}-close to a separable state.
So the probability of acceptance of \wQoFYwHz{W}
is at most $ 1 - \kappa_2 \cdot 4^{-n}
+ \frac{\kappa_1}{9} \cdot 4^{-n} $.
If $ \kappa_1 $ is sufficiently small then
$ c-s = \uaivIgOh{4^{-n}} $, so the corollary follows.
\end{proof}
\subsection{\texorpdfstring{Notes on \cdToKcBp{} Proof Systems}
{Notes on BellQMA Proof Systems}}
\label{fCJGhyBh}
An interesting consequence of
\MJqbMYrm{NmPOZGKI} and
\ref{nKVfzRHR} is that
in the small-gap setting, \pUeVeClC{k} proof
systems have the same power as \lrcGeBqo{k} proof
systems if $k$ is at least a linear function of
the input length.
Such a result is not known to hold in the
normal-gap setting.
As we mentioned before, in the normal-gap
setting we now that if $k$ is constant
then \pUeVeClC{k} collapses to \SHvwttaw,
and the proof of this fact doesn't generalize
to the small-gap setting.
This means that it is an interesting open question
to figure out the power of \pUeVeClC{2} with
exponentially small gap.
There could be two possibilities:
\begin{itemize}
\item It is quite unlikely that the small-gap
version of $\pUeVeClC{2} = \vxJDyEiW$,
because it would give a very powerful proof
system, and moreover it would show that
if a verifier is restricted to Bell-measurements
then it gains a lot of extra power if we decrease the bound on
the gap; since $\pUeVeClC{2} = \SHvwttaw$ if the gap
is inverse-polynomial, but $\pUeVeClC{2} = \lrcGeBqo{2}$
if it is inverse-exponential.
\item So we conjecture that the small-gap
$\pUeVeClC{2} \subset \vxJDyEiW$.
But this would imply that somewhere
between constant and linear number of provers,
the power of \pUeVeClC{k} significantly increases.
\end{itemize}
We leave the study of this class for future work.
\subsection{Error and Proof Reduction}
Note that as a side-product of our results,
in both the \pUeVeClC{k} and \lrcGeBqo{k} proof
systems we can `amplify' the error from
double-exponentially small gap to single-exponentially
small gap.
Also in the case of \lrcGeBqo{k}, we can make the proof
system to have \emph{one-sided error},
which up to our knowledge,
has only been shown to hold for \ACroxxZA{} \cite{Jordan2011}.
(Additionally, the number of proofs can be reduced to two,
but this also follows from \cite{Harrow2010},
where an essentially different argument was used.)
\subsection{\texorpdfstring{\lrcGeBqo{k} with Small Gap and Logarithmic-Length Proofs}
{QMA[k] with Small Gap and Logarithmic-Length Proofs}}
\label{sec:QMAlog}
We also examine multi-prover \SHvwttaw{} proof systems
with small gap and where each proof consist of at most
\FIqGvTZG{n} qubits, in the hope that we can separate them
from the ones that have \wrDoIJTw{n}-length proofs.
Unfortunately we were not able to do this in general,
only in the case where the gap is inverse-exponential,
which follows from a few simple observations.
If the gap is smaller, say double-exponentially small
or unbounded, then we give a lemma which simplifies the
proof system by converting it to one with single-qubit
proofs without changing the order of magnitude of the gap.
The details follow.
\begin{lem}
$ \displaystyle \sVwdElyU =
\bigcup_{\substack{0 < s < c \leq 1, \\ c-s \geq 2^{ - \ipeoICuS }}}
\KOzxJkVZ{\FIqGvTZG{n}}{\wrDoIJTw{n}}{c}{s} $,
where \HxLQfJAj{c}{n} and \HxLQfJAj{s}{n} can be calculated in
\wrDoIJTw{n}-time on a classical computer.
\end{lem}
\begin{proof}[Proof sketch]
The containment $ \sVwdElyU \subseteq
\bigcup_{\substack{0 < s < c \leq 1, \\ c-s \geq 2^{ - \ipeoICuS }}}
\KOzxJkVZ{\FIqGvTZG{n}}{\wrDoIJTw{n}}{c}{s} $ is trivial
since in \sVwdElyU{} the gap is always at least
inverse-exponential.
So we only need to prove the other direction.
Consider the class \KOzxJkVZ{\FIqGvTZG{n}}{\wrDoIJTw{n}}{c}{s}
for some $c$ and $s$, for which the conditions in the
lemma hold.
Note that $ \KOzxJkVZ{\FIqGvTZG{n}}{\wrDoIJTw{n}}{c}{s} \subseteq
\vtZmVJmM{c'}{s'} $, where $ c' - s' \geq 2^{ - \ipeoICuS } $.
(In fact they are equal.)
The reason behind it is that for each proof there exists
a unitary transformation that creates it, say, from \raIfJZxu{0}.
Since the unitaries are on \FIqGvTZG{n}-many qubits,
they can be described up to exponential precision
by quantum circuits of size \wrDoIJTw{n}.
The \ACroxxZA{} verifier expects these descriptions as its proof.
The completeness and soundness follows easily.\footnote{For
a detailed (but a bit different) proof of this fact
see \cite{Gharibian2011}.}
\par
The above \ACroxxZA{} proof system can be converted to an \kJmbkUXz{}
proof system with exponentially small gap in the following way.
The \kJmbkUXz{} proof is the same as the \ACroxxZA{} proof.
Once the proof is fixed the question is to estimate the
acceptance probability of a \ipeoICuS-time quantum computation
to exponential precision.
This problem is in \AIJjtsxK, and since $ \AIJjtsxK = \sVwdElyU $,\footnote{\AIJjtsxK{}
is the unbounded-gap version of \VGZgApSi.
For the exact definition of the class and the proof of the above
equality see for example the survey by Watrous \cite{Watrous2008a}.}
we have an \kJmbkUXz{} proof system with the given parameters.
This \kJmbkUXz{} class is obviously in \sVwdElyU{} via the same argument
as $ \kJmbkUXz \subseteq \sVwdElyU $ for the normal-gap version of \kJmbkUXz.
\end{proof}
\begin{rem}
Note that if we decrease the proof lengths to only $1$
qubits then the resulting class is obviously still
equals to \sVwdElyU.
More formally
\[ \displaystyle \sVwdElyU =
\bigcup_{\substack{0 < s < c \leq 1, \\ c-s \geq 2^{ - \ipeoICuS }}}
\KOzxJkVZ{1}{\wrDoIJTw{n}}{c}{s} \text{,} \]
where \HxLQfJAj{c}{n} and \HxLQfJAj{s}{n} can be calculated in
\wrDoIJTw{n}-time on a classical computer.
\end{rem}
This observation is generalized by the following lemma
to any smaller gaps.
\begin{lem}
$ \displaystyle \KOzxJkVZ{\FIqGvTZG{n}} {\ipeoICuS} {c} {s}
\subseteq \KOzxJkVZ{1} {\ipeoICuS}
{1 - 2^{- t} \cdot \DnoYoBfx{1-c}}
{1 - 2^{- t} \cdot \DnoYoBfx{1-s}} $,
for some $t$ \ygMPbVVP{} $ \HxLQfJAj{t}{n} \in \wrDoIJTw{n} $,
and $c$ and $s$ are arbitrary functions of $n$.
\label{HTTPbCwP}
\end{lem}
We present the proof of this lemma in
\IhApwfvN{zKhwWgDg}
on page~\pageref{zKhwWgDg}.
Let us explicitly state here the previously mentioned corollary.
\begin{cor}
It holds that
\begin{align*}
\bigcup_{\substack{0 < s < c \leq 1, \\ c-s \geq 2^{ - 2^{\ipeoICuS} }}}
\KOzxJkVZ{\FIqGvTZG{n}}{\ipeoICuS}{c}{s} &=
\bigcup_{\substack{0 < s < c \leq 1, \\ c-s \geq 2^{ - 2^{\ipeoICuS} }}}
\KOzxJkVZ{1}{\ipeoICuS}{c}{s} \text{,}
\intertext{and}
\bigcup_{0 < s < c \leq 1} \KOzxJkVZ{\FIqGvTZG{n}}{\ipeoICuS}{c}{<s} &=
\bigcup_{0 < s < c \leq 1} \KOzxJkVZ{1}{\ipeoICuS}{c}{<s}
\text{.}
\intertext{Note that in \dWQdKYxc{HTTPbCwP}
the one-sided error property is preserved.
So we also have for example}
\KOzxJkVZ{\FIqGvTZG{n}}{\ipeoICuS}{1}{<1} &=
\KOzxJkVZ{1}{\ipeoICuS}{1}{<1} \text{.}
\end{align*}
\end{cor}
We don't know any better upper-bound for
\KOzxJkVZ{1}{\ipeoICuS}{c}{s} with $c-s \geq 2^{ - 2^{\ipeoICuS} }$
then the trivial \vxJDyEiW.
It would be interesting to strengthen this bound
or give some non-trivial lower-bound.
We leave this question for future work.
\subsection{More Open Problems}
Here we list some more open problems that
we think may be interesting to work on.
\begin{itemize}
\item What is the power of \lrcGeBqo{k} and
\pUeVeClC{k} with unbounded gap?
Can we at least show some upper-bounds?
\item What is the power of \KOzxJkVZ{1}{1}{1}{1 - 2^{- 2^{\ipeoICuS}}}
or \KOzxJkVZ{1}{1}{1}{<1}?
\qpcpnInM{} the proof system has only one qubit
as its proof but we allow double-exponentially
small or unbounded gap.
Are they the same as \AIJjtsxK?
Note that the known
$ \KOzxJkVZ{\FIqGvTZG{n}}{1}{\frac{2}{3}}{\frac{1}{3}}
= \VGZgApSi $ proofs \cite{Marriott2005,Beigi2011} break
down if the gap is so small.
\end{itemize}
\section*{Acknowledgements}
The author would like to thank Rahul Jain and
Penghui Yao for helpful discussions on the topic.
\appendix
\section{Proof of Soundness for \MJqbMYrm{QoOYZYVy}}
\label{WEFVXHAK}
This section proves soundness for verifier
\wQoFYwHz{V} described by \FxucuSoM{jEVcePOJ}
on page~\pageref{jEVcePOJ};
and hence finishes the proof of \MJqbMYrm{QoOYZYVy}.
The proof is done through a few lemmas.
\par
From now on let us suppose that $ \oVrArWkh \notin \HNGSAaAW $,
and let's denote the state of $ \VSdxZJNO{R}_1 $ by \raIfJZxu{\psi},
and the state of $ \VSdxZJNO{R}_2 $ by \raIfJZxu{\varphi}.
These two states can be written in the form
\[ \raIfJZxu{\psi} = \sum_{i=0}^{2^n - 1} \alpha_i \raIfJZxu{i}
\sum_{j=0}^2 \beta_{i,j} \raIfJZxu{j} \text{,} \qquad
\raIfJZxu{\varphi} = \sum_{i=0}^{2^n - 1} \alpha_i' \raIfJZxu{i}
\sum_{j=0}^2 \beta_{i,j}' \raIfJZxu{j} \text{,} \]
where $ \sum_i \AfBMaxZI{\alpha_i}^2 = \sum_i \AfBMaxZI{\alpha_i'}^2 = 1 $, and
for all $i$, $ \sum_j \AfBMaxZI{\beta_{i,j}}^2 = \sum_j \AfBMaxZI{\beta_{i,j}'}^2 = 1 $.
\par
The following lemma says that if the Equality Test succeeds with high
probability then the distribution of outcomes in
the Consistency Test will be similar.
This is analogous to Lemma~2.5 of \cite{Blier2009}.
\begin{lem}
If the Equality test of \FxucuSoM{jEVcePOJ}
succeeds with probability at least $1 - \varepsilon$,
then for all $k$ and $\ell$ it holds that
$ \AfBMaxZI{ \AfBMaxZI{\alpha_k \beta_{k,\ell}}^2 -
\AfBMaxZI{\alpha_k' \beta_{k,\ell}'}^2 } \leq
\sqrt{8 \varepsilon} $.
\label{TAMDFIWI}
\end{lem}
\begin{proof}
Let $ p_{i,j} = \AfBMaxZI{\alpha_i \beta_{i,j}}^2 $ and
$ q_{i,j} = \AfBMaxZI{\alpha_i' \beta_{i,j}'}^2 $,
and let us denote the probability vector
with elements $p_{i,j}$ by $p$, and similarly for $q$.
Then we have the following.
\begin{align}
\sqrt{1 - \AfBMaxZI{\NChmtXsR{\psi}{\varphi}}^2}
&= \vPiROeRG{\GyFczGjW{\psi}}{\GyFczGjW{\varphi}}
\label{Bhpmbwna} \\
&\geq \frac{1}{2} \xOJdCehH{p - q}
\label{YtqvQXzj} \\
&= \frac{1}{2} \sum_{i,j} \AfBMaxZI{ \AfBMaxZI{\alpha_i \beta_{i,j}}^2
- \AfBMaxZI{\alpha_i' \beta_{i,j}'}^2 } \nonumber \\
&\geq \frac{1}{2} \AfBMaxZI{ \AfBMaxZI{\alpha_k \beta_{k,\ell}}^2
- \AfBMaxZI{\alpha_k' \beta_{k,\ell}'}^2 } \text{,} \nonumber
\end{align}
for any $k$ and $\ell$.
Equations \eqref{Bhpmbwna}
and \eqref{YtqvQXzj}
are from the properties of the trace distance.
\MJqbMYrm{QzwgCixS} implies that
$ \frac{1}{2} \DnoYoBfx{1 + \AfBMaxZI{\NChmtXsR{\psi}{\varphi}}^2}
\geq 1 - \varepsilon $.
With the above derivation the claim of the lemma follows.
\end{proof}
Similarly to Lemma~2.6 of \cite{Blier2009}
(or Lemma~6.1 of \cite{Chiesa2011}), the next lemma states that
vertices with high probability of being observed
have a well-defined color.
\begin{lem}
Suppose that \raIfJZxu{\psi} and \raIfJZxu{\varphi} pass
the Equality Test of \FxucuSoM{jEVcePOJ}
and also line~\ref{DkvnMEbW}
in the Consistency Test with
probability at least $ 1 - 10^{-10} \cdot 4^{-n} $.
Then for all $i$ for which
$ \AfBMaxZI{\alpha_i}^2 \geq 100^{-1} \cdot 2^{-n} $,
there exist one $j$ for which $ \AfBMaxZI{\beta_{i,j}}^2 \geq 0.9 $.
\label{bXmyQJDh}
\end{lem}
\begin{proof}
Towards contradiction suppose that $ \exists i $ \ygMPbVVP{}
$ \AfBMaxZI{\alpha_i}^2 \geq \frac{1}{100 \cdot 2^{n}} $ and
$ \forall j $ it holds that
$ \AfBMaxZI{\beta_{i,j}}^2 < \frac{9}{10} $.
Then \bmZlBfAS{} we can say that
$ \AfBMaxZI{\beta_{i,0}}^2 \geq \frac{1}{20} $ and
$ \AfBMaxZI{\beta_{i,1}}^2 \geq \frac{1}{20} $.
Since the probability that the Equality Test succeeds
is at least $ 1 - \frac{1}{10^{10} \cdot 4^n}$,
we can apply \dWQdKYxc{TAMDFIWI} and get that
\[ \AfBMaxZI{ \AfBMaxZI{\alpha_i \beta_{i,1}}^2 - \AfBMaxZI{\alpha_i' \beta_{i,1}'}^2 }
\leq \frac{\sqrt{8}}{ 10^5 \cdot 2^n } . \]
This implies that
\begin{align*}
\AfBMaxZI{\alpha_i'}^2 \AfBMaxZI{\beta_{i,1}'}^2
&\geq \AfBMaxZI{\alpha_i}^2 \AfBMaxZI{\beta_{i,1}}^2
- \frac{\sqrt{8}}{10^5 \cdot 2^n} \\
&\geq \frac{1}{2000 \cdot 2^{n}}
- \frac{\sqrt{8}}{10^5 \cdot 2^n} \text{.}
\end{align*}
Then the probability that in line~\ref{YKlinzMJ}
in the Consistency Test we get
$ v_1 = v_2 = i $, $c_1 = 0$ and $c_2 = 1$ is
\[ \TdSeQgKh{v_1=i \text{ and } c_1=0} \cdot \TdSeQgKh{v_2=i \text{ and } c_2=1}
\geq \frac{1}{2000 \cdot 2^{n}}
\DnoYoBfx{\frac{1}{2000 \cdot 2^{n}} - \frac{\sqrt{8}}{10^5 \cdot 2^n}}
> \frac{1}{10^{10} \cdot 4^n} . \]
This contradicts to the assumption that
\raIfJZxu{\psi} and \raIfJZxu{\varphi} pass
line~\ref{DkvnMEbW} with
probability at least $ 1 - 10^{-10} \cdot 4^{-n} $.
\end{proof}
The next lemma is analogous to Lemma~2.7 of \cite{Blier2009}
and also to Lemma~6.2 of \cite{Chiesa2011}.
\begin{lem}
Suppose that \raIfJZxu{\psi} and \raIfJZxu{\varphi} pass
the Equality Test of \FxucuSoM{jEVcePOJ}
and also line~\ref{DkvnMEbW}
in the Consistency Test with
probability at least $ 1 - 10^{-10} \cdot 4^{-n} $.
Then the probability of measuring $0$ on
$\VSdxZJNO{C}_1$ in line~\ref{GTSlDeLX}
in the Uniformity Test is at least $0.05$.
\label{KnUWPcsP}
\end{lem}
\begin{proof}
Suppose that we measure $\VSdxZJNO{N}_1$ in the
standard basis.
If the outcome is $i$ then the probability of measuring $0$ on
$\VSdxZJNO{C}_1$ with the measurement of
\MOVPSHbF{NlcnHZqK} is
\[ \frac{1}{3} \AfBMaxZI{ \beta_{i,0} + \beta_{i,1} + \beta_{i,2} }^2 . \]
For all $i$ for which $ \AfBMaxZI{\alpha_i}^2 \geq \frac{1}{100 \cdot 2^{n}} $
\dWQdKYxc{bXmyQJDh} applies, which means that
there exist $k_i$ \ygMPbVVP{} $ \AfBMaxZI{\beta_{i,k_i}}^2 \geq \frac{9}{10} $.
Let $ \ell_i \faJHPjTb k_i + 1 \mod 3 $ and
$ m_i \faJHPjTb k_i + 2 \mod 3 $;
then $ \AfBMaxZI{\beta_{i,\ell_i}}^2 + \AfBMaxZI{\beta_{i,m_i}}^2 < \frac{1}{10} $.
We can lower bound the above probability by
\begin{align*}
\frac{1}{3} \AfBMaxZI{ \beta_{i,k_i} + \beta_{i,\ell_i} + \beta_{i,m_i} }^2 &\geq
\frac{1}{3} \AfBMaxZI{ \AfBMaxZI{\beta_{i,k_i}} - \AfBMaxZI{\beta_{i,\ell_i} + \beta_{i,m_i}} }^2 \\
&\geq \frac{1}{3} \DnoYoBfx{ \AfBMaxZI{\beta_{i,k_i}} -
\sqrt{ 2 \cdot \DnoYoBfx{ \AfBMaxZI{\beta_{i,\ell_i}}^2 + \AfBMaxZI{\beta_{i,m_i}}^2 }} }^2 \\
&\geq \frac{1}{3} \DnoYoBfx{ \frac{9}{10} - \sqrt{\frac{2}{10}} }^2 \\
&> \frac{6}{100} \text{,}
\end{align*}
where the second inequality follows from the \vHtwIXhr{}.
Or more precisely we have
$ \AfBMaxZI{\beta_{i,\ell_i} + \beta_{i,m_i}}^2 \leq
2 \cdot \DnoYoBfx{ \AfBMaxZI{\beta_{i,\ell_i}}^2 + \AfBMaxZI{\beta_{i,m_i}}^2 } <
\frac{2}{10} < \frac{9}{10}
\leq \AfBMaxZI{\beta_{i,k_i}}^2 $.
The probability of measuring $0$ on
$\VSdxZJNO{C}_1$ in line~\ref{GTSlDeLX} is
\begin{align*}
\sum_{i=0}^{2^n-1} \AfBMaxZI{\alpha_i}^2 \cdot
\frac{1}{3} \cdot \AfBMaxZI{ \beta_{i,0} + \beta_{i,1} + \beta_{i,2} }^2
&\geq \sum_{\substack{i \text{ for which} \\
\AfBMaxZI{\alpha_i}^2 \geq \frac{1}{100 \cdot 2^n} }}
\AfBMaxZI{\alpha_i}^2 \cdot
\frac{1}{3} \cdot \AfBMaxZI{ \beta_{i,0} + \beta_{i,1} + \beta_{i,2} }^2 \\
&> \frac{6}{100} \cdot \sum_{\substack{i \text{ for which} \\
\AfBMaxZI{\alpha_i}^2 \geq \frac{1}{100 \cdot 2^n} }}
\AfBMaxZI{\alpha_i}^2 \\
&\geq \frac{6}{100} \cdot
\DnoYoBfx{1 - \frac{2^n-1}{100 \cdot 2^n}} \\
&> \frac{5}{100} \text{,}
\end{align*}
where we used the fact that at most $2^n - 1$ nodes ($i$'s) can have
$ \AfBMaxZI{\alpha_i}^2 < \frac{1}{100 \cdot 2^n} $.
\end{proof}
In order to proceed we need two lemmas, one from
\cite{Blier2009} and one from \cite{Chen2010}.
We present them now, together with their proofs.
\begin{lem}[Lemma~2.8 of \cite{Blier2009}]
Let us have a state $\raIfJZxu{\xi} \in \VWqEPttf^m$,
$ \raIfJZxu{\xi} = \sum_{i=0}^{m-1} \gamma_i \raIfJZxu{i} $.
If there exist a $k$ \ygMPbVVP{} $ \AfBMaxZI{\gamma_k}^2 < \frac{1}{2m} $,
then the probability of getting $1$ when we measure
\raIfJZxu{\xi} with the measurement of \MOVPSHbF{NlcnHZqK}
is at least $ \frac{1}{16 m^2} $.
\label{tAaagiyx}
\end{lem}
\begin{proof}
Let $p$ and $q$ be the probability distributions that
arise when we measure \raIfJZxu{\xi} and \raIfJZxu{u_m} in the
computational basis.
Or in other words, let $p$ be the probability vector with
elements $\AfBMaxZI{\gamma_i}^2$, and $q$ be the vector with
all elements equal to $\frac{1}{m}$.
Similarly to \dWQdKYxc{TAMDFIWI}
we have that
\begin{align*}
\sqrt{1 - \AfBMaxZI{\NChmtXsR{u_m}{\xi}}^2}
&= \vPiROeRG{\GyFczGjW{u_m}}{\GyFczGjW{\xi}} \\
&\geq \frac{1}{2} \xOJdCehH{q - p} \\
&= \frac{1}{2} \sum_{i=0}^{m-1} \AfBMaxZI{ \frac{1}{m}
- \AfBMaxZI{\gamma_i}^2 } \\
&\geq \frac{1}{2} \AfBMaxZI{ \frac{1}{m} - \AfBMaxZI{\gamma_k}^2 } \\
&> \frac{1}{4m} \text{.}
\end{align*}
Since the probability of getting $1$ when we measure
\raIfJZxu{\xi} with the measurement of
\MOVPSHbF{NlcnHZqK} is
$ 1 - \AfBMaxZI{\NChmtXsR{u_m}{\xi}}^2 $, the
statement of the lemma follows.
\end{proof}
The following argument appears in the proof of Lemma~3
of \cite{Chen2010}, which we state here as a separate lemma.
\begin{lem}
Suppose that we have a bipartite quantum state
$ \raIfJZxu{\psi} \in \zdLEuPyh{N} \NUxdODWZ \zdLEuPyh{C} $,
with $ \zdLEuPyh{N} = \VWqEPttf^N $ and $ \zdLEuPyh{C} = \VWqEPttf^C $.
We can write this state as
\[ \raIfJZxu{\psi} = \sum_{i=0}^{N-1} \alpha_i \raIfJZxu{i}
\sum_{j=0}^{C-1} \beta_{i,j} \raIfJZxu{j} \text{,} \]
where $ \sum_i \AfBMaxZI{\alpha_i}^2 = 1 $, and
for all $i$, $ \sum_j \AfBMaxZI{\beta_{i,j}}^2 = 1 $.
Suppose that the probability of measuring $0$ on
\zdLEuPyh{C} with the measurement of
\MOVPSHbF{NlcnHZqK} is $p$,
and after the measurement the resulting state on
\zdLEuPyh{N} is
\[ \raIfJZxu{\xi} \in \zdLEuPyh{N} \text{,} \qquad
\raIfJZxu{\xi} = \sum_{i=0}^{N-1} \gamma_i \raIfJZxu{i} \text{,} \]
with $ \sum_i \AfBMaxZI{\gamma_i}^2 = 1 $.
Then for all $i$ it holds that
\[ \AfBMaxZI{\alpha_i}^2 \geq p \cdot \AfBMaxZI{\gamma_i}^2 . \]
\label{KHySvfFa}
\end{lem}
\begin{proof}
Let $q_i$ denote the probability that
if we measure the \zdLEuPyh{C} part of \raIfJZxu{\psi}
with the measurement of \MOVPSHbF{NlcnHZqK}
we get outcome $0$, then if we measure the \zdLEuPyh{N}
part in the standard basis, we get $i$.
Note that for all $i$,
\[ q_i = p \cdot \AfBMaxZI{\gamma_i}^2 . \]
On the other hand,
\begin{align*}
q_i &= \fHAwcqZo{\psi}
\DnoYoBfx{ \GyFczGjW{i} \NUxdODWZ \GyFczGjW{u_C} } \raIfJZxu{\psi} \\
&= \frac{\AfBMaxZI{\alpha_i}^2}{C} \cdot
\AfBMaxZI{ \sum_{j=0}^{C-1} \beta_{i,j} }^2 \\
&\leq \frac{\AfBMaxZI{\alpha_i}^2}{C} \cdot C \cdot
\sum_{j=0}^{C-1} \AfBMaxZI{\beta_{i,j}}^2 \\
&= \AfBMaxZI{\alpha_i}^2 \text{,}
\end{align*}
where the inequality above follows from the \vHtwIXhr.
The above derivations imply the statement of the lemma.
\end{proof}
Analogously to Lemma~2.9 of \cite{Blier2009}
(and to Lemma~6.3 of \cite{Chiesa2011}), the following
lemma says that if the states pass some of the tests with high probability,
then it must be that all nodes appear with high enough probability.
\begin{lem}
Suppose that \raIfJZxu{\psi} and \raIfJZxu{\varphi} pass
the Equality Test of \FxucuSoM{jEVcePOJ},
line~\ref{DkvnMEbW}
in the Consistency Test,
and also the Uniformity Test with
probability at least $ 1 - 10^{-10} \cdot 4^{-n} $.
Then for all $i$,
$ \AfBMaxZI{\alpha_i}^2 \geq 100^{-1} \cdot 2^{-n} $.
\label{mkiZykmS}
\end{lem}
\begin{proof}
Because of \dWQdKYxc{KnUWPcsP}
the probability of measuring $0$ on
$\VSdxZJNO{C}_1$ in line~\ref{GTSlDeLX}
of \FxucuSoM{jEVcePOJ} is at least $\frac{5}{100}$.
Let the state of $\VSdxZJNO{N}_1$ be
$ \raIfJZxu{\xi} = \sum_{i=0}^{2^n-1} \gamma_i \raIfJZxu{i} $,
after we got $0$ on $\VSdxZJNO{C}_1$.
Towards contradiction suppose that $ \exists i $
\ygMPbVVP{} $ \AfBMaxZI{\alpha_i}^2 < \frac{1}{100 \cdot 2^n} $.
Since we got this measurement result with probability
$ \geq \frac{5}{100} $, \dWQdKYxc{KHySvfFa}
implies that $ \AfBMaxZI{\gamma_i}^2 < \frac{1}{5 \cdot 2^n} $.
From \dWQdKYxc{tAaagiyx}
the probability of getting $1$ when measuring
$\VSdxZJNO{N}_1$ in line~\ref{GTSlDeLX}
is at least $ \frac{1}{16 \cdot 4^n} $.
So the probability of failing the Uniformity Test is at least
$ \frac{5}{100} \cdot \frac{1}{16 \cdot 4^n}
> \frac{1}{10^{10} \cdot 4^n} $.
This is a contradiction.
\end{proof}
The following lemma finishes the proof of soundness for
verifier \wQoFYwHz{V}.
\begin{lem}[Soundness]
If $ \oVrArWkh \notin \HNGSAaAW $ then verifier
\wQoFYwHz{V} described by \FxucuSoM{jEVcePOJ}
will reject with probability at least
$ \displaystyle \frac{1}{3 \cdot 10^{10} \cdot 4^n} $.
\label{SaalgorQ}
\end{lem}
\begin{proof}
Assume that \raIfJZxu{\psi} and \raIfJZxu{\varphi} pass
the Equality Test, the Uniformity Test, and
line~\ref{DkvnMEbW} of
\FxucuSoM{jEVcePOJ} with
probability at least $1 - \frac{1}{10^{10} \cdot 4^n}$,
as otherwise we are done.
Let \HxLQfJAj{c}{i} be equal to the $j$ for which \AfBMaxZI{\beta_{i,j}} is
maximal, or in other words,
\[ \HxLQfJAj{c}{i} \faJHPjTb \argmax_j \AfBMaxZI{\beta_{i,j}} . \]
Because of \dWQdKYxc{mkiZykmS} and
\ref{bXmyQJDh} this maximum
is well defined.
According to \dWQdKYxc{mkiZykmS},
when measuring \raIfJZxu{\psi} in line~\ref{YKlinzMJ},
the probability of obtaining \DnoYoBfx{k, \HxLQfJAj{c}{k}}, for all $k$,
is at least $ \AfBMaxZI{\alpha_k}^2 \cdot \frac{9}{10}
\geq \frac{1}{100 \cdot 2^n} \cdot \frac{9}{10}
> \frac{1}{120 \cdot 2^n} $.
Similarly, from \dWQdKYxc{TAMDFIWI},
for all $k$ the probability that we get \DnoYoBfx{k, \HxLQfJAj{c}{k}}
when measuring \raIfJZxu{\varphi} in
line~\ref{YKlinzMJ}, is at least
$ \frac{1}{120 \cdot 2^n} - \frac{\sqrt{8}}{10^5 \cdot 2^n}
> \frac{1}{240 \cdot 2^n} $.
Since the graph is not 3-colorable
$ \exists u,v \in V $ \ygMPbVVP{}
$ \DnoYoBfx{u,v} \in E $ and $ \HxLQfJAj{c}{u} = \HxLQfJAj{c}{v} $.
If in line~\ref{YKlinzMJ} we get
\DnoYoBfx{u, \HxLQfJAj{c}{u}} and \DnoYoBfx{v, \HxLQfJAj{c}{v}} then
the Consistency Test will reject.
This happens with probability at least
\[ \frac{1}{120 \cdot 2^n} \cdot
\frac{1}{240 \cdot 2^n}
> \frac{1}{10^{10} \cdot 4^n} . \]
Since the Consistency Test is chosen with probability
$\frac{1}{3}$, the statement of the lemma follows.
\end{proof}
\section{Proof of \dWQdKYxc{HTTPbCwP}}
\label{zKhwWgDg}
This section presents the proof of
\dWQdKYxc{HTTPbCwP}
from page~\pageref{HTTPbCwP}.
We will need the following previously known facts.
\begin{lem}[See \eRAJAvVL{} Chapter~4.5.2 of \cite{Nielsen2000}]
An arbitrary unitary operator on $m$ qubits
can be implemented using a circuit
containing \QUPOtzml{m^2 4^m} single-qubit and \gctqElrH{} gates.
\label{grWbYQoV}
\end{lem}
\begin{lem}
Any two-dimensional unitary operator \nVzyhtCm{U}
can be written in the form
\[ \nVzyhtCm{U} = e^{\WHYFCIFc \theta} \sYcexCeY{z}{\alpha} \lDGygnkY
\sYcexCeY{z}{\beta} \lDGygnkY \sYcexCeY{z}{\gamma} \text{,} \]
for $ \alpha, \beta, \gamma, \theta \in
\left[ 0, 2 \pi \right) $.
\label{XPxqbsml}
\end{lem}
\begin{proof}
It is well-known that one can write any \nVzyhtCm{U} as
$ \nVzyhtCm{U} = e^{\WHYFCIFc \vartheta} \sYcexCeY{z}{\alpha}
\sYcexCeY{x}{\beta} \sYcexCeY{z}{\gamma} $.\footnote{See
\eRAJAvVL{} Theorem~4.1 and Exercise~4.11 from \cite{Nielsen2000}.}
The lemma follows from the fact that
$ \sYcexCeY{x}{\beta} = e^{\WHYFCIFc \xi}
\lDGygnkY \sYcexCeY{z}{\beta} \lDGygnkY $.
\end{proof}
The following definition and lemma are from \cite{Bravyi2005},
but we restate them in a different way,
and include the proof for the sake of completeness.
\begin{defi}
Define the `magic' state \raIfJZxu{m_{\omega}} to be
$ \displaystyle \raIfJZxu{m_{\omega}} \faJHPjTb \frac{1}{\sqrt{2}}
\DnoYoBfx{ \raIfJZxu{0} + e^{\WHYFCIFc \omega} \raIfJZxu{1} } $.
\end{defi}
\begin{lem}
There exists a constant size quantum circuit,
made up of \lDGygnkY, \gctqElrH{} and classical logical gates,
\ygMPbVVP{} on input \raIfJZxu{\varphi} and \raIfJZxu{m_{\omega}},
it produces the state $ \sYcexCeY{z}{\omega} \raIfJZxu{\varphi} $
with probability $ 1/2 $.
(This probability is independent of \raIfJZxu{\varphi}
and $\omega$.)
\label{APGlxUTv}
\end{lem}
\begin{proof}
Let $ \raIfJZxu{\varphi} = a \raIfJZxu{0} + b \raIfJZxu{1} $
for some $ a,b \in \VWqEPttf $, with
$ \AfBMaxZI{a}^2 + \AfBMaxZI{b}^2 = 1 $.
On input $ \raIfJZxu{m_{\omega}} \NUxdODWZ \raIfJZxu{\varphi} $
the circuit performs the projective measurement defined
by projectors
\[ \nVzyhtCm{P}_1 = \GyFczGjW{00} + \GyFczGjW{11} \text{,} \qquad
\nVzyhtCm{P}_2 = \GyFczGjW{01} + \GyFczGjW{10} \text{.} \]
Note that
\begin{align*}
\nVzyhtCm{P}_1 \DnoYoBfx{ \raIfJZxu{m_{\omega}} \NUxdODWZ \raIfJZxu{\varphi} } &=
\frac{1}{\sqrt{2}} \DnoYoBfx{ a \raIfJZxu{00} + e^{\WHYFCIFc \omega} b \raIfJZxu{11} }
\text{, and} \\
\nVzyhtCm{P}_2 \DnoYoBfx{ \raIfJZxu{m_{\omega}} \NUxdODWZ \raIfJZxu{\varphi} } &=
\frac{1}{\sqrt{2}} \DnoYoBfx{ b \raIfJZxu{01} + e^{\WHYFCIFc \omega} a \raIfJZxu{10} }
\text{.}
\end{align*}
This means that both measurement outcomes
happen with probability $1/2$.
If the outcome is $2$, the circuit
fails to create the state.
However, if the measurement outcome is $1$,
then it performs a \gctqElrH{},
which disentangles the second qubit, and gets the state
\[ a \raIfJZxu{0} + e^{\WHYFCIFc \omega} b \raIfJZxu{1} =
\sYcexCeY{z}{\omega} \raIfJZxu{\varphi} \text{.} \qedhere \]
\end{proof}
By combining the lemmas above we get the following corollary.
\begin{cor}
Let \nVzyhtCm{U} be any unitary operator on $m$ qubits.
There exist an algorithm \wQoFYwHz{A}, using
\lDGygnkY, \gctqElrH{} and classical logical gates, that takes as input
the classical description of the
circuit representing \nVzyhtCm{U}
(as in \dWQdKYxc{grWbYQoV} and
\ref{XPxqbsml})
and the unentangled magic states
(as in \dWQdKYxc{APGlxUTv}), and
produce the state $ \nVzyhtCm{U} \raIfJZxu{0} $,
with probability $ 2^{- \QUPOtzml{m^2 4^m}} $.
Both the length of the input and the running time of
\wQoFYwHz{A} is \QUPOtzml{m^2 4^m}.
\label{OzBzXiMl}
\end{cor}
We are ready to prove
\dWQdKYxc{HTTPbCwP}.
\begin{proof}[Proof of \dWQdKYxc{HTTPbCwP}]
Let $ L \in \KOzxJkVZ{\ell}{k}{c}{s} $, where
$ \ell \in \FIqGvTZG{n} $ and
$ k \in \wrDoIJTw{n} $.
(As usual, $n$ is the length of the input.)
\cuIynkaR{} assume that all the proofs have length $\ell$.
Let \wQoFYwHz{V} be the corresponding verifier.
We now construct a proof system
that recognizes the same language $L$, and
all the proofs are $1$ qubit long.
\par
Denote the new verifier by \wQoFYwHz{W}.
For its proofs, \wQoFYwHz{W} expects to get the circuit descriptions of
$k$ unitary operators
$ \nVzyhtCm{U}_1 $, \ldots, $ \nVzyhtCm{U}_k $, together with
the corresponding unentangled $1$ qubit magic states.
Let us denote the number of magic states by $t$.
Since each $\nVzyhtCm{U}_i$ lives on $\ell$ qubits,
the total number of bits and qubits \wQoFYwHz{W} gets
from the provers are at most
$ \QUPOtzml{k \cdot \ell^2 4^{\ell}} \in \wrDoIJTw{n} $.
\wQoFYwHz{W} uses algorith \wQoFYwHz{A} from
\dOYmqxlU{OzBzXiMl}, $k$ times, to get the states
$ \nVzyhtCm{U}_1 \raIfJZxu{0} $, \ldots, $ \nVzyhtCm{U}_k \raIfJZxu{0} $.
If in any of these $k$ cases \wQoFYwHz{A} fails then
\wQoFYwHz{W} accepts. Otherwise \wQoFYwHz{W} runs \wQoFYwHz{V}
on $ \nVzyhtCm{U}_1 \raIfJZxu{0} \NUxdODWZ \cdots
\NUxdODWZ \nVzyhtCm{U}_k \raIfJZxu{0} $
and accepts \iZJgfzde{} \wQoFYwHz{V} accepts.
The running time of \wQoFYwHz{W} is obviously
polynomial.
Let us denote the probability of acceptance of
\wQoFYwHz{V} on input $ \nVzyhtCm{U}_1 \raIfJZxu{0} \NUxdODWZ \cdots
\NUxdODWZ \nVzyhtCm{U}_k \raIfJZxu{0} $ by $p$.
There are two cases for \wQoFYwHz{W}.
\begin{itemize}
\item With probability $ 1 - 2^{-t} $ one of the runs of algorithm
\wQoFYwHz{A} fails and \wQoFYwHz{W} accepts.
\item With probability $ 2^{-t} $ all the \wQoFYwHz{A}s
succeed. In this case \wQoFYwHz{W} obtains the state
$ \nVzyhtCm{U}_1 \raIfJZxu{0} \NUxdODWZ \cdots
\NUxdODWZ \nVzyhtCm{U}_k \raIfJZxu{0} $.
Given this state as input to \wQoFYwHz{V},
it accepts with probability $p$.
\end{itemize}
The overall probability with which \wQoFYwHz{W} accepts is
\[ 1 - 2^{-t} + 2^{-t} p = 1 - 2^{-t} \DnoYoBfx{1-p} \text{.} \]
\par
We are left to argue about completeness and soundness.
For the completeness assume that there exist states
$ \raIfJZxu{\varphi_1} $, \ldots, $ \raIfJZxu{\varphi_k} $,
with which \wQoFYwHz{V} accepts with probability
at least $c$.
Then there exists unitary operators
$ \nVzyhtCm{U}_1' $, \ldots, $ \nVzyhtCm{U}_k' $ \ygMPbVVP{}
$ \raIfJZxu{\varphi_i} = \nVzyhtCm{U}_i' \raIfJZxu{0} $ for
$ i \in \iteYiYHH{k} $.
The honest provers of \wQoFYwHz{W} can give the descriptions of
these unitaries together with the corresponding magic states,
so the completeness parameter follows.
For the soundness assume that for all states
$ \raIfJZxu{\varphi_1} $, \ldots, $ \raIfJZxu{\varphi_k} $,
\wQoFYwHz{V} accepts with probability at most $s$.
Then even if \wQoFYwHz{W} gets arbitrary circuit descriptions
and magic states, they correspond to some
unitaries when all the runs of \wQoFYwHz{A} succeed.
So the soundness follows similarly to the completeness.
\end{proof}
\LeqcErMT
\end{document}